\documentclass[journal]{IEEEtran}
\ifCLASSINFOpdf
 
\else

\fi
\usepackage{amsmath,amssymb}
\makeatletter
\newsavebox\myboxA
\newsavebox\myboxB
\newlength\mylenA

\newcommand*\xoverline[2][0.75]{%
	\sbox{\myboxA}{$\m@th#2$}%
	\setbox\myboxB\null
	\ht\myboxB=\ht\myboxA%
	\dp\myboxB=\dp\myboxA%
	\wd\myboxB=#1\wd\myboxA
	\sbox\myboxB{$\m@th\overline{\copy\myboxB}$}
	\setlength\mylenA{\the\wd\myboxA}
	\addtolength\mylenA{-\the\wd\myboxB}%
	\ifdim\wd\myboxB<\wd\myboxA%
	\rlap{\hskip 0.5\mylenA\usebox\myboxB}{\usebox\myboxA}%
	\else
	\hskip -0.5\mylenA\rlap{\usebox\myboxA}{\hskip 0.5\mylenA\usebox\myboxB}%
	\fi}
\makeatother
\usepackage{cite}
\usepackage{varioref}
\usepackage{amsmath}
\usepackage[usenames, dvipsnames]{color}

\usepackage{graphicx}
\usepackage[ruled]{algorithm2e}
\usepackage{caption}
\usepackage{threeparttable}
\usepackage{amsthm}
\usepackage{theoremref}
\makeatletter
\newcommand{\removelatexerror}{\let\@latex@error\@gobble}
\makeatother
\newtheorem{claim}{Claim}

\newtheorem{rule1}{Rule}

\newtheorem{lemma}{Lemma}

\newtheorem{Definition}{Definition}

\usepackage{amsmath}

\begin{document}

\title{A $1$-approximation algorithm for energy-efficient TDM-PON guaranteeing SLA of up-stream and down-stream traffic}

\author{Sourav~Dutta, Dibbendu~Roy,~and~Goutam~Das
\thanks{Sourav Dutta is with the Department
of Electronics and Electrical Communication Engineering, Indian institute of Technology Kharagpur,  Kharagpur,
India (e-mail: sourav.dutta.iitkgp@gmail.com).}
\thanks{Dibbendu Roy and Goutam Das are with G. S. Sanyal School of Telecommunication, Indian Institute of Technology Kharagpur, Kharagpur, India (e-mail: dibbaroy@gmail.com, gdas@gssst.iitkgp.ernet.in).} 
}
\maketitle
\date{\vspace{-5ex}}
\vspace{-2cm}
\begin{abstract}
Economical and environmental concerns necessitate research on designing energy-efficient optical access network especially Ethernet Passive Optical Network (EPON) which is one of the most widely accepted and deployed last-mile access network. In this paper, our primary focus is on designing a protocol for saving energy at Optical Network Units (ONUs) while satisfying the Service Label Agreement (SLA). The SLA of both Up-Stream (US) and Down-Stream (DS) traffic can be satisfied only if the EPON network can react to their instantaneous load change during sleep periods of ONUs and to the best of our knowledge, there doesn't exist any such proposal. Towards this target, we propose a mechanism that allows the Optical Line Terminal (OLT) to force ONUs to wake-up from sleep mode. Here, we demonstrate that if the OLT can distribute the active ONUs (transceivers are active) fairly among cycles then it provides a significant improvement in energy-efficiency. To achieve this, we formulate an ILP for fairly distributing active ONUs among cycles while satisfying the SLA of both US and DS traffic at the same time. A polynomial time $1$-approximation algorithm is proposed for solving this ILP. The convergence and the complexity analysis of the algorithm are also performed. Extensive simulations depict that fair distribution of ONUs reduces the power consumption and average delay figure at the same time and the reduction increases with an increment of the number of ONUs and round-trip time.  
\end{abstract}

\IEEEpeerreviewmaketitle
\section{Introduction}  \label{intro}
The Ethernet Passive Optical Network (EPON), being one of the most widely appreciated and deployed  access network,  energy-efficient EPON design has already become a well-established research area \cite{kani2013power}.  An EPON network architecture contains an Optical Line Terminal (OLT), multiple Optical Network Units (ONUs) and Remote nodes \cite{ipact}. Traffic from the OLT to ONUs and from ONUs to the OLT are termed as Down-stream (DS) and Up-stream (US) traffic respectively.  ONUs are responsible for consuming $70\%$ of the overall EPON power consumption \cite{internetpowerconsumption}, which promotes the research on energy-efficient ONU design. For saving energy at ONU, several low power modes (e.g. deep sleep (ds), fast sleep (fs), cyclic sleep (cs), doze mode (dz)) with different power consumption figures and sleep-to-wake-up time have been proposed \cite{internetpowerconsumption}. While saving energy by employing any of these Low Power Modes (LPMs), an ONU must ensure the Service Label Agreements (SLAs) (for example, delay bounds or packet drop probability). Thus, an efficient protocol design is required in order to choose the most suitable LPM and the duration over which the LPM (sleep duration) is employed while preserving SLA of both US and DS traffic. All these protocols are broadly classified as OLT-assisted and ONU-assisted.

In an OLT-assisted protocol, the OLT takes decision about LPMs and sleep duration which are then informed to ONUs through extra information within the GATE message while ONUs are oblivious about the decision making process.  
Several OLT-assisted protocols have been proposed in the literature \cite{oltass0,oltass1,oltass2,oltass3,oltass4,oltass5,cyclicsleep}. In Multi Point Control Protocol (MPCP), which is the standardized MAC protocol for EPON network \cite{ieee2010ieee}, ONUs inform their buffer size to the OLT through REPORT message. Once an ONU enters into sleep mode, it can neither receive the GATE message not send REPORT message.  Thus, the OLT-assisted protocols cannot react to the instantaneous change of both US and DS traffic which is extremely important for serving bursty traffic arrivals. This may lead to violation of SLAs which is a major drawback of all OLT-assisted protocols.          
In an ONU-assisted protocol, ONUs individually decide their sleep modes and sleep durations by following their own protocols and the OLT is completely unaware about this process. Several ONU-assisted protocols are also present in the literature \cite{onuass0,onuass1,doze,chayan,chayanjournal}. 
If an ONU enters into a sleep mode, it can observe its own US traffic arrivals but not DS traffic. Thus, ONU-assisted protocols are capable of reacting to the instantaneous change of only the US traffic. This may lead to violation of SLAs for DS traffic which is a major drawback of ONU-assisted protocols. 

In all existing sleep mode protocols, the sleep duration of ONUs are decided independently and hence, multiple ONUs may wake up in one polling cycle leading to bandwidth crunch. However, if these ONUs wake-up in different cycles and the entire bandwidth is distributed among the active ONUs (ONUs whose transceivers are switched on), then ONUs can up-stream a certain number of packets in shorter duration. This will reduce the time period over which an ONU is in active mode (active duration) and it provides an opportunity of improving the energy-efficiency significantly. Thus, instead of deciding sleep durations of ONUs independently, if their sleep durations are decided such that when they wake-up from sleep mode, minimum number of ONUs remain active, then energy-efficiency can be improved. This requires fair distribution of active ONUs among polling cycles.    
Thus, the major drawbacks of all existing proposals can be summarized as:
\begin{itemize}
	\item None of the existing protocols can react to the instantaneous load variation of both US and DS traffic at the same time.
	\item Sleep durations of different ONUs are decided independently and hence, the possibility of improving energy-efficiency by fairly distributing them among cycles remains unexploited. 
\end{itemize}

In this paper, we try to resolve these two issues. During sleep period, ONUs can observe the US traffic while the OLT can observe the DS traffic. Hence, while deciding sleep durations, a co-operation of the OLT and ONUs  presents an opportunity to maintain SLAs of both US and DS traffic at the same time. It is understood that ONUs can wake-up from sleep mode for plausible SLA violations in the US. However, in traditional EPON architecture, the OLT is incapable of forcing ONUs to wake-up from sleep mode when the OLT observes a plausible SLA violation in the DS. If the OLT is provisioned with the facility of forcing ONUs to wake-up from sleep mode, then the SLA of both US and DS traffic could be maintained at the same time. Further, this facility also allows the OLT to allocate a cycle to ONUs, when they will up-stream after waking-up from sleep, such that allocation is fair and the SLA of US and DS is maintained.

The discussions made above lead us to following major objectives: Firstly, the OLT needs to be provisioned with a facility of waking-up ONUs from sleep mode. Secondly, once OLT is provisioned with the wake-up facility, it has to decide which of the ONUs to be awakened in a cycle such that the active ONUs are fairly distributed among cycles while satisfying SLAs for both US and DS traffic. Lastly, the OLT has to decide the time instants when the ONUs (chosen by the previous decision process) are to be awakened from their sleep modes. To the best of our knowledge, this paper is the first proposal which aims at meeting the mentioned objectives, allowing ONUs to sleep while ensuring the SLA of both US and DS traffic at the same time. The contributions of this paper in their order of appearance are:

\begin{itemize}
	\item We present a mechanism which allows the OLT to force ONUs to wake-up from their sleep modes.
	\item To decide on which of the ONUs are to be awakened in a cycle, we formulate an Integer Linear Program (ILP) with the mentioned objectives. To solve the ILP, a polynomial time 1-approximation algorithm \cite{approx} named Fair Distribution of ONUs among Slots (\textit{FDOS}) is proposed. The convergence and complexity of this algorithm is also analyzed.
	\item It presents a scheduling protocol for deciding the time instants when the OLT forces the ONUs to wake-up from sleep modes. We have named this as \textit{wake-up scheduling}.
	\item  Through extensive simulations, we demonstrate a significant improvement in energy-efficiency, as compared to all existing protocols. Further, the simulations exhibit increase in improvement of energy-efficiency on enhancement of both the number of ONUs and the round-trip time which indicate that the proposed  \textit {FDOS} algorithm and \textit{wake-up scheduling} protocol is suitable for both ultra-dense and rural network.
\end{itemize}

The rest of the paper is organized as follows. In Section II, we briefly describe the required background of this paper. A mechanism for providing provision to the OLT of waking-up ONUs from sleep mode is proposed in Section III. In Section IV, we formulate the ILP and describe the wake-up scheduling protocol. The proposed FDOS algorithm is proposed in Section V. The energy-efficiency and average delay figures of the proposed mechanism are presented in Section VI. In Section VII, concluding statements are provided.      

\section{Background}\label{bac}
In this section, we first provide the literature review and then briefly describe our previously proposed ONU-assisted protocol, named as OSMP-EO. 
\subsection{Literature Survey}
The authors of \cite{oltass0} have proposed Sleep Mode Aware (SMA) protocol where before sending every GATE message, the OLT predicts the minimum value of time instant, before which the next GATE message of that ONU will not be sent. This time interval is then informed to the ONU during which it employs LPMs. Effect of employing different LPMs is demonstrated in \cite{oltass1}. The same authors have shown the effect of employing different schemes for calculating the transmission slots for both US and DS traffic in \cite{oltass2}.  In SMA protocol, the sleep duration is decided in every cycle and hence, it doesn't affect the delay characteristic. However, the sleep duration being insignificant, the energy saving figure is quite low. In \cite{oltass3}, the authors have shown that a significant improvement in energy-efficient can be achieved by increasing the packet delay.  The Green Bandwidth Allocation (GBA) Algorithm, proposed in \cite{oltass4}, is mathematically analyzed by modeling buffers of every ONU as M/G/1 queue with vacation \cite{oltass5}. In \cite{cyclicsleep}, the authors have proposed a new sleep mode, watchful sleep mode, which employs the advantages of both cyclic sleep and doze mode. The major drawback of \cite{oltass3,oltass4,oltass5,cyclicsleep} is that, since the ONUs sleep for multiple polling cycles and ONUs don't participate in the decision making process, they cannot react to the instantaneous load change of both US and DS traffic leading to plausible violation of SLAs in US and DS.

The authors of \cite{onuass0} have proposed a protocol where an ONU observes the DS traffic for a certain period. If no DS traffic arrives over this period then it sleeps for a fixed duration. The same authors have extended this protocol for both US and DS traffic in \cite{onuass1}. Since, in these two protocols, sleep durations are fixed, they cannot respond to an instantaneous load change of both US and DS traffic.  However, an ONU can observe its own US traffic arrivals even  during the sleep period and hence, it is capable of reacting to the instantaneous change of the US traffic, which has been considered in \cite{doze,chayan,chayanjournal}. The protocols, proposed in \cite{doze} consider doze mode for saving energy. In \cite{chayan}, we have demonstrated that a significant improvement in energy saving can be achieved by selecting a suitable LPM instead of using a single one. This protocol has been extended for delay-sensitive traffic in \cite{chayanjournal}. 
In \cite{dutta}, we have proposed an ONU-assisted mechanism for applying doze mode during the time periods of the active cycles when the other ONUs up-stream and proposed a new ONU-assisted protocol named as OSMP-EO. In this paper, we consider OSMP-EO as the ONU assisted protocol followed by the ONUs to save energy. The reasoning and corresponding details behind considering OSMP-EO are discussed in Section \ref{ssec:algo_frem}. We now briefly explain the OSMP-EO protocol. 
\subsection{OSMP-EO}\label{sssec:osmp}
Here, we briefly describe our previously proposed OSMP-EO protocol. The protocol has been designed for delay-insensitive traffic where the buffer sizes of ONUs are limited  and it attempts to reduce the energy consumption of ONUs while avoiding the packet drop. In OSMP-EO, ONUs alter between deep sleep ($ds$), fast sleep ($fs$) and active mode ($on$). Let, at the current time $t$, $ONU_i$ is in sleep mode $S_m (\in \{ds,fs\})$. In this case, the next decision is taken after a fixed time $T_m$ (i.e at $t_1=t+T_m$), when $ONU_i$ decides whether to retain the sleep mode $S_m$ or it wakes-up. We now explain the rules that are followed by $ONU_i$ for taking this decision. At $t_1$, if sleep mode $S_m$ is retained then the wake-up process can be initiated after $T_m$ duration and the wake-up process requires $T_{sw}^{S_m}$ duration ($T_{sw}^{S_m}$- Sleep-to-wake-up time for sleep mode $S_m$). The OSMP-EO being an ONU-assisted protocol, $ONU_i$ doesn't have any information about the GATE message arrival time. Thus, there is a possibility that $ONU_i$ becomes active immediately after the arrival of GATE message and hence, has to wait until  the arrival of the next GATE message for reporting the buffer state. The bandwidth will be granted in the next cycle. Therefore, if $S_m$ is retained at $t_1$, then the packet drop can be avoided if the buffer doesn't fill up in next $T_m+T_{sw}^{S_m}+2T_{cm}$ duration ($T_{cm}$-maximum cycle time). So, at $t_1$, $ONU_i$ first predict the buffer fill-up time ($T_{BU}^i$) and decide its mode ($M_i$) by following \thref{rule1}.
\begin{rule1}\thlabel{rule1}
	If $T_{BU}^i>T_m+T_{sw}^{S_m}+2T_{cm}$, $S_m$ is retained 
	and otherwise, the wake-up process is initiated. 
\end{rule1}
$ONU_i$ follows \thref{rule1} in every $T_m$ durations until it decides to wake-up from $S_m$. Let, at $t_2$, $ONU_i$ decides to wake-up from $S_m$ which takes $T_{sw}^{S_m}$ duration and therefore, $ONU_i$ becomes active (i.e $M_i=on$) at $t_3=t_2+T_{sw}^{S_m}$. The next decision about $M_i$ will be taken only after up-streaming all packets that are in the buffer at $t_3$.  
In \cite{dutta}, we have shown that a sleep mode provides enhancement in energy-efficiency as compared to other sleep modes with higher power consumption figure, if the buffer fill-up time (which indicates sleep duration) is more than a threshold. Let us denote this threshold for $ds$ and $fs$ as $T_{lb}^{ds}$ and $T_{lb}^{fs}$ respectively and the value of them is calculated in \cite{dutta}.
Thus, $ONU_i$ follows \thref{rule2} for deciding $M_i$.
\begin{rule1} \thlabel{rule2}	
	If $T_{lb}^{ds}\leq T_{BU}^i$ then $M_i=ds$ while if $T_{lb}^{fs}<T_{BU}^i\leq T_{lb}^{ds}$ then $M_i=fs$ and otherwise,
	$M_i=on$ 	
\end{rule1} 
In a polling based protocol like MPCP, the transmitter of an ONU remain idle when other ONUs up-stream. Thus, during the cycles when $ONU_i$ is active, the transmitter of $ONU_i$ can be switched-off during other ONUs' US transmission. For doing this in a complete ONU-assisted manner, we have proposed a mechanism in \cite{dutta}. Next, we discuss the architectural modifications at the OLT which would enable it to wake-up ONUs.
\section{Architectural Modifications}
From the previous discussion (refer Section \ref{intro}), it is now clear that if the SLA of both US and DS data is needed to be satisfied then the OLT should have the provision to force ONUs to wake-up from sleep mode. The OLT has to send some extra information to ONUs in order to wake them from their sleep modes. It is well understood that this information can be realized by a single bit which indicates whether an ONU should wake up or not. This extra information can be sent either  through the GATE message or a new message, termed as "wake-up message". However, during sleep periods, ONUs are unable to receive this wake-up information as the receiver is switched-off. A simple way of resolving this issue is that ONUs periodically switch-on their receivers and then handshake with the OLT for receiving the wake-up information \cite{cyclicsleep}. Then, it is important to decide this period of waking-up the receiver. If this time period is large, then the possibility of violation of SLA of the DS data will be huge which is not desirable. Whereas, a shorter value of it reduces energy-efficiency as the receivers are kept active over a significantly large time period just to receive this one-bit wake-up information. However, since the wake-up information is just one bit, it can be sent with a very low data rate (much lower as compared to the usual data transmission). Thus, detection of the wake-up information doesn't require the entire ONU receiver and a simple ON-OFF Keying (OOK) detection circuit is sufficient to do it. Thus, if an OOK detection circuit (very low power consumption figure) is included at ONUs and they are kept active during the sleep periods of ONUs then any time the OLT can force ONUs to wake-up from sleep mode. In traditional TDM-PON architecture \cite{ipact}, the wake-up information of a particular ONU reaches to all ONUs. Thus, it is important to design a mechanism by which an ONU can identify its own wake-up message. In order to do this, we have used the concept of Sub-Carrier Multiplexing (SCM) \cite{SCM,SCM1}. 

In our proposed mechanism, the OLT allocates $N$ different RF frequencies (sub-carriers) to all $N$ ONUs by which the wake-up message is transmitted, whereas the actual carrier is used for transmitting the DS data. We now explain our propose ONU architecture with the help of Fig. \ref{fig:arec}. The Photo Detector (PD) receives all sub-carriers ($f_{c_1}$, $f_{c_1}$, $\dots$, $f_{c_N}$) along with the down-stream wavelength. The received signal is then divided into two parts. One part is passed through a Low Pass Filter (LPF) which filters out the DS data which is then processed through the line-card (normal PON reception). The other portion is passed through a Band-Pass Filter (BPF) which receives the sub-carrier, allocated for it. The sub-carrier is then frequency down-converted and then detected by OOK detector or in other word, the sub-carrier is detected by a ASK detector. If it receives a wake-up message then the wake-up process of the actual receiver will be initiated. Thus, our proposed architectural modifications allow the OLT to force the ONUs to wake-up from sleep mode at any time instant. 
It can be noted that the required bandwidth at sub-carrier is very low and hence, many sub-carriers can be created.       
\begin{figure}[t]
	\centering
	\includegraphics[scale=.35]{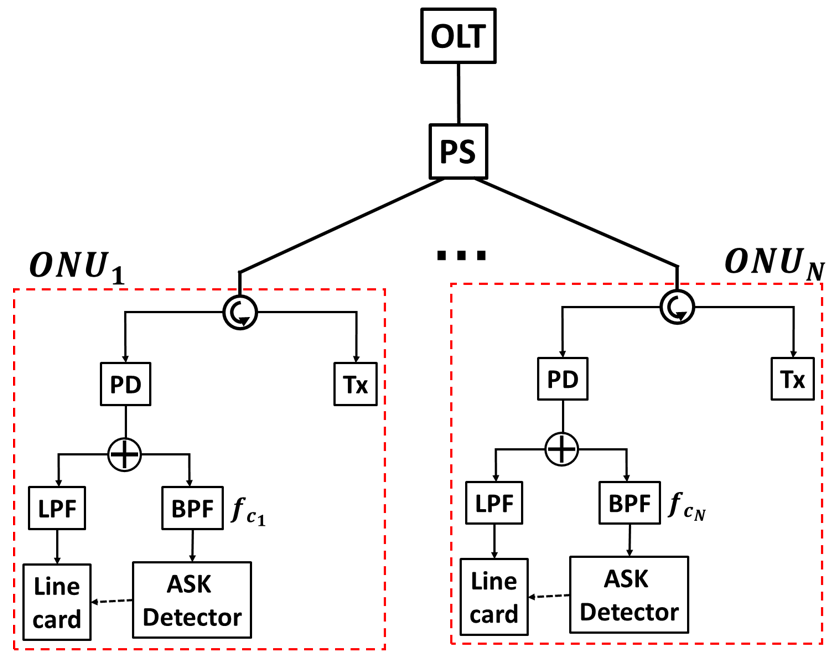}
	\caption{Proposed architecture}
	\label{fig:arec}
\end{figure}     
\section{Optimization problem formulation and wake-up scheduling}
Here, we first provide the algorithm framework of our proposed mechanism. Thereafter, we formulate an optimization problem to decide on which of the ONUs are to be awakened in a cycle which is then used for wake-up scheduling.  
\subsection{Algorithmic framework}\label{ssec:algo_frem}
In this paper, we consider that the optimization problem for deciding which ONUs are to be awakened in a cycle is formulated at the beginning of every polling cycle. It can be noted that, for DBA \cite{ipact}, the duration of polling cycles (cycle times) are adaptive and at a certain time instant, the OLT doesn't have the information about future cycle times. Thus, we divide future times into multiple slots of size, same as the maximum cycle time ($T_{cm}$). Let $\mathbb{N}_s$ and $\mathbb{S}_n$ denote the set of ONUs that are in sleep mode and the set of all slots respectively (in next subsection we show that cardinality of $\mathbb{S}_n$ i.e. $|\mathbb{S}_n|$ is finite). So, the optimization problem allocates a slot ($\in \mathbb{S}_n$) to every ONU ($\in \mathbb{N}_s$) and the objective is to maximize the energy-efficiency.

In a sleep mode protocol, an ONU alters between sleep mode and active mode. 
If the active periods can be reduced then it automatically increases sleep durations which provide an improvement in energy-efficiency. Thus, in this paper, our main objective is to minimize the total active periods of ONUs. One possible way to achieve this is by fairly distributed active ONUs among cycles (refer section \ref{intro}) which causes reduction of the number of active ONUs in a cycle. Since the OLT itself decides which ONUs will be active in a cycle, it can distribute the entire cycle only among the active ONUs. Therefore, in a cycle, an active ONU can up-stream more amount of US data resulting in a reduction of active periods.  Further, we know that every time an ONU wakes-up from sleep mode, a certain time instant is wasted when ONUs consume full power but no US transmission is possible. Thus, another possible way to reduce the total active period of an ONU is by reducing the possibility of state transition from sleep mode to active mode which can be achieved 
by keeping the ONU in the sleep mode as long as possible, after entering into it. While maximizing energy-efficiency, the OLT should maintain SLA (for example, delay bound, packet drop etc.) of both US and DS traffic. In this paper, for simplicity, we consider only the packet drop as the measure of SLA. We also show that a few minor modifications allow  including the average delay bound constraint as well. Thus, for maximizing  energy-efficiency while avoiding packet-drop of both US and DS traffic, in this paper, we first formulate an optimization problem and then for solving this problem, a polynomial time 1-approximation algorithm, namely FDOS, is proposed. 

The OLT can always observe the DS traffic, the OLT predicts  the buffer fill-up time of $ONU_i$ ($\forall i$), $T_{BD}^i$ and schedules the wake-up message such that $ONU_i$ initiates DS transmission before the OLT buffer gets filled up. However, the US traffic arrivals of an ONU can be observed only by that ONU during sleep periods of ONUs. So, in this work, we consider that an ONU (say $ONU_i$) predicts its buffer fill-up time for the US traffic ($T_{BU}^i$) and inform it along with the decided sleep mode ($S_m$) immediately before entering into $S_m$ to the OLT through the REPORT message. This process requires modification of the MPCP. However, the OLT gets some idea about the US traffic arrivals with the information of the report size which can be utilized to estimate $T_{BU}^i$ and hence, $S_m$ of $ONU_i$. In this process, no modification of MPCP is required. However, in this case, the prediction error will be more. In this paper, we use the first process. In can be noted that during sleep mode, the instantaneous traffic load may change. Hence, to avoid packet drop of US traffic, ONUs are provided the facility to wake-up from sleep mode as well and we assume that the previously proposed OSMP-EO protocol is followed by ONUs. Our proposed mechanism is the following:
\begin{itemize}
	\item ONUs follow the OSMP-EO protocol and immediately before entering a sleep mode $S_m$, they inform the predicted buffer fill-up time $T_{BU}^i$ and sleep mode $S_m$ through the REPORT message.
	\item Using these information, at the beginning of every polling cycles, the OLT  first runs the FDOS algorithm for deciding which of the ONUs are to be awakened in a cycle and then the wake-up scheduling is performed.
	\item  The IPACT protocol \cite{ipact} is used for scheduling the GATE message while the Limited scheme (LS) is used as a grant-sizing protocol \cite{ipact}. In LS, grant-size of $ONU_i$ ($G_i$) is same as report size ($R_i$) but not more than a threshold ($T_m^i$) i.e. $G_i=\max(R_i,T_m^i)$. However, the value of $T_m^i$ $\forall i$ is calculated at the beginning of every polling cycles immediately after the FDOS algorithm by:
	$$T_m^i=\begin{cases}
	\dfrac{B_m}{N_a} ~~ \text{if the polling cycle is allocated to $ONU_i$}\\
	0  ~~~~~~\text{otherwise}
	\end{cases} $$ 
	where $B_m$ and $N_a$ denotes the maximum allowable bandwidth of a cycle and the number of active ONUs of that cycle respectively. 
\end{itemize}

\begin{table}[t]
	\begin{center}
		\caption{{Definition of notations}}
		\begin{tabular}{|c|p{7cm}|}
			\hline
			\multicolumn{1}{|c|}{\textbf{Notation}} & 
			\multicolumn{1}{c|}{\textbf{Description}}\\ \hline
			$\mathbb{N}_s$ & Set of all ONUs that are in sleep mode\\
				$\mathbb{N}_s$ & Set of all slots at which ONUs of $\mathbb{S}_n$ can be assigned\\
				$|\mathbb{M}_j|$ & Set of all ONUs that are assigned of slot $j$\\
				$|\mathbb{N}_s|$ & Cardinality of set $\mathbb{N}_s$\\
				$|\mathbb{S}_n|$ & Cardinality of set $\mathbb{S}_n$\\
				$|\mathbb{M}_j|$ & Cardinality of set $\mathbb{M}_j$\\
				$b$ & $\Big\lceil\dfrac{|\mathbb{N}_s|}{|\mathbb{S}_n|}\Big\rceil$\\
				$x^*_{ij}$ & Optimal assignment \\
				$\hat{x}_{ij}$ & Assignment obtained from FDOS algorithm. \\
				$n^*_j$ & $\sum\limits_ix^*_{ij}$ \\
				\hline 
		\end{tabular}
	\end{center}
\end{table} 
\subsection{Problem formulation}
 Here we formulate the proposed optimization problem. 
 \subsubsection{Objective} 
 Since one objective is to maximize fairness, we have to first define a quantitative measurement of fairness. In this paper, we use Jain's Fairness Index \cite{fairind} as the measure of fairness and it is defined as:
 \begin{align}
 J(x_1,x_2,\dots,x_n)=\dfrac{\Big(\sum\limits_{i=0}^nx_i\Big)^2}{\sum\limits_{i=0}^nx_i^2}
 \end{align} 
 In our case, the ONUs are fairly assigned (makes the ONU active) over $| \mathbb{S}_n|=M$ number of cycles. Here, $| \mathbb{S}_n|$ denotes cardinality of set $ \mathbb{S}_n$. If $n_j$ denotes the number of active ONUs in the $j^{th}$ slot then in our case, the Jain's Index is:
\begin{align}
	J(n_1,n_2,\dots,n_M)=\dfrac{\Big(\sum\limits_{j=0}^{M-1}n_j\Big)^2}{\sum\limits_{j=0}^{M-1}n_j^2}
\end{align} 
It is evident that an increment of $J(n_1,n_2,\dots,n_M)$ implies the allocation is more fair. So, one objective of the optimization problem  is to maximize $J(n_1,n_2,\dots,n_M)$.  Since exactly one slot ($\in \mathbb{S}_n$) is assigned to all $| \mathbb{N}_s|=N$ ONUs, $\sum\limits_{j=0}^{M-1}n_j=N$ and hence, maximization of $J(n_1,n_2,\dots,n_M)$ is same as minimization of $\sum\limits_{j=0}^{M-1}n_j^2$. Let us define a binary variable $x_{ij}$ where 
\begin{align}
x_{ij}=
\begin{cases}
1 &\text{if $j^{th}$ slot is assigned to $ONU_i$}\\
0 &\text{otherwise}
\end{cases}
\end{align} It is easy to note that $n_j=\sum\limits_{i=0}^{N-1}x_{ij}$ and $n_j^2=\sum\limits_{i=0}^{N-1}\sum\limits_{k=0}^{N-1}x_{ij}x_{kj}$.
So, fair distribution of ONUs ($\in \mathbb{N}_s$) among slots ($\in \mathbb{S}_n$) require minimization of 
$\sum\limits_{j=0}^{M-1}\sum\limits_{i=0}^{N-1}\sum\limits_{k=0}^{N-1}x_{ij}x_{kj}=f_1(x_{ij})$.

The other objective of the optimization problem is to maximize the sleep duration of ONUs. Thus, the allocated slot should be as late as possible or in other word, maximization of $\sum\limits_{j=0}^{M-1}jx_{ij}~\forall i$ is required. Maximization of sleep durations for all ONUs can be mathematically modeled by taking weighted ($w_i$ is the weight for $ONU_i$) summation of the objectives of all ONUs and the objective function is given by: $\sum\limits_{i=0}^{N}\sum\limits_{j=0}^{M-1}jw_ix_{ij}$ or in general, $\sum\limits_{i=0}^{N-1}\sum\limits_{j=0}^{M-1}w_{ij}x_{ij}=f_2(x_{ij})$ where $w_{ij}=jw_i$ in this case. The optimization problem is multi-objective where the objective functions are: $f_1(x_{ij})$ and $f_2(x_{ij})$. An weighted average of these two objectives can be taken for making the multi-objective optimization to a single objective optimization. If $W$ is the weight factor then the objective function of the optimization, $f(x_{ij})$, can be defined as:   
\begin{align}\label{mini}
f(x_{ij})=Wf_1(x_{ij})-f_2(x_{ij})
\end{align}  
Since, our primary objective is to maximize the fairness, in this paper, we assume $W$ is very high: $W>\sum\limits_{i=0}^{N-1}\sum\limits_{j=0}^{M-1}w_{ij}$.
\subsubsection{Constraints}
Now, we explain all constraints that are needed to be satisfied while minimizing $f(x_{ij})$. 
\paragraph{\textbf{Single slot allocation constraint}}
As mentioned in Section \ref{ssec:algo_frem}, a single slot is assigned to every ONUs ($\in \mathbb{N}_s$) which can be ensured by providing the constraint given in eq. (\ref{const1}).
\begin{align}\label{const1}
\sum\limits_{j=0}^{M-1}x_{ij}=1~~~~\forall i\in\mathbb{N}_s
\end{align}

Next, we find all feasible slots that can be allocated to $ONU_i ~(\forall i\in \mathbb{N}_s)$. 
\paragraph{\textbf{Wake-up constraint due to past wake-up message transmission}}
The set $\mathbb{N}_s$ also includes ONUs for whom the wake-up messages have already been sent and the set of these ONUs is denoted as $\mathbb{N}_{ws}$. let, at time $t_{w}^i$, the wake-up message has been sent to $ONU_i$ by the OLT which takes $\frac{T_{rtt}^i}{2}$ time interval to reach the $ONU_i$. Then, it initiate the wake-up process from sleep mode $S_m^i$ which requires another $T_{sw}^{S_m^i}$ and thus, $ONU_i$ can receive GATE message after $t_{w}^i+\frac{T_{rtt}^i}{2}+T_{sw}^{S_m^i}$ time instant. Let, the last GATE message for $ONU_i$ has been sent at $t_{lg}^i$. $ONU_i$ will be able to receive this GATE if $t_{w}^i+T_{sw}^{S_m^i}<t_{lg}^i$ and in this case, in the first slot only, $ONU_i$ will be able to send the REPORT and initiate the US transmission in the next slot. In other cases, the OLT must be able to receive the REPORT message in $\Big\lceil\dfrac{t_{w}^i+T_{rtt}^i+T_{sw}^{S_m^i}}{T_{cm}}\Big\rceil=S_a-1$ slot and the OLT allocates grant to $ONU_i$ in the next slot i.e. $S_a$. This can be ensured by the constraint equation, given in eq. (\ref{wake_sent}).
\begin{align}\label{wake_sent}
	x_{ij}=
	S_a ~~\forall i\in\mathbb{N}_{ws} 
\end{align}    

We now show that for all other ONUs (i.e ONUs that are in set $\mathbb{N}_s\setminus\mathbb{N}_{ws}$), the slot number of all feasible slots are in between an upper and lower bound. 
\subsubsection*{\textbf{Upper bound calculation}}
While allocating slots, the OLT should avoid packet drop of both US and DS traffic (refer section \ref{intro} which also provides two constraints namely US packet drop avoidance constraint and DS packet drop avoidance constraint. Further,  in MPCP \cite{ieee2010ieee}, if an ONU doesn't send any US data over a certain period then the OLT de-register that ONU. Avoidance of de-registration imposes an upper-bound on sleep duration and hence, the slot number of feasible slots which is termed as De-registration avoidance constraint.   

\paragraph{\textbf{US packet drop avoidance constraint}}
Let, at time instant $t_{lr}^i$, the OLT received the latest REPORT message from $ONU_i$ when $ONU_i$ informed the buffer fill-up time and sleep mode as $T_{BU}^i$ and $S_m^i$ respectively. So, the latest REPORT was transmitted by $ONU_i$ at time instant $t_{lr}^i-\frac{T_{rtt}^i}{2}$.
If the prediction is correct then $ONU_i$ will wake-up from sleep mode $S_m^i$ at time instant $t_{lr}^i+T_{UB}^i-\frac{T_{rtt}^i}{2}-2T_{cm}$ (refer \thref{rule1} of OSPM-EO protocol, explained in section \ref{bac}). So, the OLT can only force $ONU_i$ to wake-up from $S_m^i$ before time instant, $t_{lr}^i+T_{UB}^i-\frac{T_{rtt}^i}{2}-2T_{cm}$. If $ONU_i$ wakes-up such that it can transmit US data even at the start of a slot then in that slot $ONU_i$ will definitely be able to send the REPORT message. The US transmission will be initiated in the next slot. Thus, in this case, $ONU_i$ must send the REPORT message on or before  $\Big\lfloor \dfrac{t_{lr}^i+T_{UB}^i-2T_{cm}-t}{T_{cm}}\Big\rfloor=UB_{US}^i-1$ slot (as shown in Fig. \ref{fig:ub_us}) and up-stream data in  $UB_{US}^i$ slot where $t$ is the present time instant.
\begin{figure}[t]
	\centering
	\includegraphics[scale=.4]{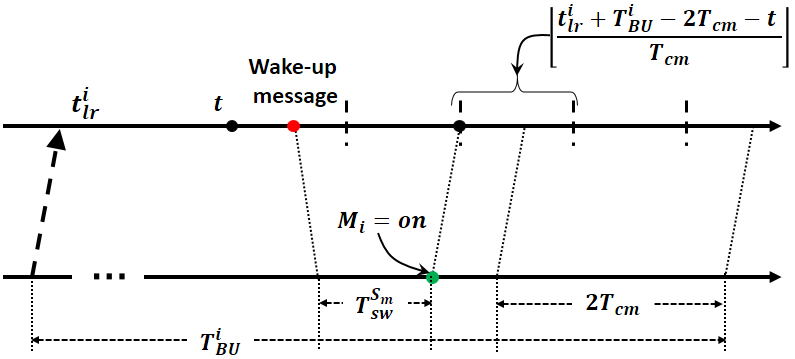}
	\caption{Calculation of upper bound of slot for avoiding packet drop of US data.}
	\vspace{-.5 cm}
	\label{fig:ub_us}
\end{figure}
This constraint can be modeled by eq. (\ref{ub_us}). 
\begin{align}\label{ub_us}
	x_{ij}=0 \text{~~if }j>UB_{US}^i ~~\forall i\in\mathbb{N}_s\setminus\mathbb{N}_{ws} 
\end{align}
\paragraph{\textbf{DS packet drop avoidance constraint}}
Let, at current time $t$, the buffer fill-up time is predicted as $T_{BD}^i$. Thus, the DS transmission should be initiated before time instant: $t+T_{BD}^i$. This can be ensured if $ONU_i$ wakes-up from sleep mode at the beginning of the slot $\Big\lfloor\dfrac{T_{BD}^i}{T_{cm}}\Big\rfloor$. $ONU_i$ will also be able to sent the REPORT message in that slot and US transmission can be initiate in the next slot: $\Big\lfloor\dfrac{T_{BD}^i}{T_{cm}}\Big\rfloor+1=UB_{DS}^i$. Thus, the buffer fill-up time provides the following constraint, given in eq. (\ref{ub_ds}).
\begin{align}\label{ub_ds}
	x_{ij}=0 \text{~~if }j>UB_{DS}^i~~\forall i\in\mathbb{N}_s\setminus\mathbb{N}_{ws}
\end{align} 

\paragraph{\textbf{De-registration avoidance constraint}}  If an ONU sleeps for such a time period that no US data or REPORT message is send within the de-registration time then that ONU have to wait till the next contention window \cite{ieee2010ieee} when it can again register and can initiate the US transmission. In this case, ONUs cannot avoid the packet drop as they has to wait for the next contention window  which is not desirable. Thus,  an ONU should wake-up such that it can send at-least one REPORT message within this de-registration time denoted as $T_{dr}$. Since the OLT has receive the latest REPORT message from $ONU_i$ at $t_{lr}^ i$, the OLT should receive the REPORT message before $t_{lr}^ i+T_{dr}$. This can be ensured if $ONU_i$ wakes-up such that it can transmit REPORT at the beginning of the slot $\Big\lfloor\dfrac{t_{lr}^i+T_{dr}-t}{T_{cm}}\Big\rfloor=UB_{d}^i-1$ and in the next slot i.e. $UB_{d}^i$, US data transmission will be initiated. Thus, in order to avoid de-registration, the following constraint, given in eq. (\ref{ub_dreg}), is need to be satisfied will performing the assignment.  
\begin{align}\label{ub_dreg}
x_{ij}=0 \text{~~if }j>UB_{d}^i~~\forall i\in\mathbb{N}_s\setminus\mathbb{N}_{ws}
\end{align} 

Thus, eq. (\ref{ub_ds})\textendash eq.(\ref{ub_dreg}) can be combined by eq. (\ref{ub}).
\begin{align}\label{ub}
\nonumber x_{ij}=0 \text{~~if }j>\min(UB_{US}^i,UB_{DS}^i,UB_d^i)=UB^i\\\forall i\in\mathbb{N}_s\setminus\mathbb{N}_{ws}
\end{align}
Thus, the number of slots over which the allocation of ONUs ($\in \mathbb{N}_s$) can be performed i.e $|\mathbb{S}_n|$ is finite and upper bounded by $UB^i$. 
\subsubsection*{\textbf{Lower bound calculation}}
Now, we calculate the lower bound of slot number of all feasible slots.
\paragraph{\textbf{Minimum initialization time constraint}}
If the wake-up message is send at the present time, $t$ then it reaches the OLT at $t+\frac{T_{rtt}^i}{2}$ when the wake-up process will start which takes another $T_{sw}^{S_m^i}$ duration. The US transmission takes another $\frac{T_{rtt}^i}{2}$ duration to reach the OLT. Thus, $ONU_i$ cannot send the REPORT message before slot $\Big\lfloor\dfrac{t+{T_{rtt}^i}+T_{sw}^{S_m^i}}{T_{cm}}\Big\rfloor$ if the GATE message is sent after waking-up which is not always guaranteed. However, in the next slot, it will definitely be able to report its buffer state. Thus, $ONU_i$ can always be able to send the REPORT message in any slot whose slot number is not lower than  $\Big\lceil\dfrac{t+{T_{rtt}^i}+T_{sw}^{S_m^i}}{T_{cm}}\Big\rceil=LB_{p}^i-1$. In the next slot i.e. $LB_{p}^i$ US data transmission can be initiated. This provide the following constraint as given in eq. (\ref{lb_present}). 
\begin{align}\label{lb_present}
 x_{ij}=0 \text{~~if }j<LB_{p}^i\forall i\in\mathbb{N}_s\setminus\mathbb{N}_{ws}
\end{align} 

\paragraph{\textbf{Minimum sleep duration constraint}}
As discussed in OSMP-EO protocol, if $ONU_i$ enters into sleep mode $S_m^i$, an improvement in energy-efficiency can only be achieved if the buffer fill-up time is more than a threshold, $T_{lb}^{S_m^i}$ (refer section \ref{sssec:osmp}) Further, in OSMP-EO, ONUs wake-up from sleep mode $T_{sw}^{S_m^i}+2T_{cm}$ and hence, $ONU_i$ should sleep for at least $T_{lb}^{S_m^i}-T_{sw}^{S_m^i}-2T_{cm}$ or in other word, it wakes-up from sleep mode after $(t_{lr}^i-T_{rtt}^i)+T_{lb}^{S_m^i}-2T_{cm}$. By using the same argument as provided above, this event can be ensured by the following the constraint as provided in eq. (\ref{lb_minsleep}).
\begin{align}\label{lb_minsleep}
\nonumber x_{ij}=0 \text{~~if }j<LB_{ms}^i\forall i\in\mathbb{N}_s\setminus\mathbb{N}_{ws}  \\\text{ where } LB_{ms}^i=\Big\lceil\dfrac{t_{lr}^i+T_{lb}^{S_m^i}-2T_{cm}}{T_{cm}}\Big\rceil+1 
\end{align}  

Therefore, eq. (\ref{lb_present})\textendash eq. (\ref{lb_minsleep}) can be combined by eq. (\ref{lb}).
\begin{align}\label{lb}
x_{ij}=0 \text{~~if }j<\max(LB_p^i,LB_{ms}^i)=LB^i~~\forall i\in\mathbb{N}_s\setminus\mathbb{N}_{ws}
\end{align}  
Let, $\mathcal{A}$ be the set of all $(i,j)$ pairs where $ONU_i$ ($\forall i\in \mathbb{N}_s$) can be assigned to the $j^{th}$ ($\forall j\in\mathbb{S}_n$) slot. Thus, $\mathcal{A}$ is given by eq. (\ref{aset}).
\begin{align}\label{aset}
	\mathcal{A}=\{(i,j)|LB^i\leq j\leq UB^i,\forall i\in\mathbb{N}_s\}
\end{align} 
\subsubsection*{\textbf{Final optimization problem}} By summarizing eq. (\ref{mini})\textendash eq. (\ref{aset}), the optimization problem is given by eq. (\ref{optm}).
\begin{subequations}\label{optm}
	\begin{alignat}{2}
	&\!\min_{(i,j)\in\mathcal{A},~
	(k,j)\in\mathcal{A}}        &\qquad& W\sum\limits_{i,j,k}x_{ij}x_{kj}-    \sum\limits_{i,j}w_{ij}x_{ij}\label{eq:optProb}\\
	&\text{subject to} &      & \sum\limits_{j=0}^{M-1}x_{ij}=1~~~~\forall i\in\mathbb{N}_s\label{eq:constraint1}
	\end{alignat}
\end{subequations}

 If the delay bound or some other SLA parameters are considered then it modifies the set $\mathcal{A}$ while the optimization problem of eq. (\ref{optm}) remains the same. 
\subsection{Wake-up scheduling}
The OLT gets allocation for all ONUs by solving eq. (\ref{optm}). Suppose, $j^{th}$ slot is allocated to $ONU_i$. If the current time is $t$ then start time of the $j^{th}$ slot is  $t+(j-1)T_{cm}$. If the sleep mode of $ONU_i$ is $S_m^i$ then the wake-up process requires $T_{sw}^{S_c^i}+T_{rtt}^i$ duration. Thus, the wake-up message for $ONU_i$ should be sent at $t_{wk}^i=t+(j-1)T_{cm}-T_{sw}^{S_c^i}-T_{rtt}^i$. It can be noted that if $t_{wk}^i>t+T_{cm}$ then in the next decision point $ONU_i$ remains in sleep mode. Thus, the allocation of $ONU_i$ can also be performed at the beginning of the next slot without deteriorating its performance. In-fact it may provide better allocation. Thus, if $t_{wk}<t+T_{cm}$ then only the wake-up message will be sent. 
 
Now, we design an $1$-approximation algorithm for solving this optimization problem of eq. (\ref{optm}). 
\section{Proposed algorithm}
Here, we explain the proposed FDOS algorithm in detail.
\subsection{Motivation and framework}
Here, we propose a heuristic algorithm, named as FDOS (FDOS- Fair Distribution of ONUs among Slots), for solving the proposed ILP (refer eq. (\ref{optm})) and prove that even in worse case, the deviation of the value of  $f(x_{ij})$ (refer eq. (\ref{mini})), obtained by solving the algorithm from its optimal value (say $f^*(x_{ij})$) is upper bounded. Such algorithms are named as approximation algorithm and the mathematical definition of $\epsilon$-approximation algorithm \cite{approx} is given in \thref{def1}.
\begin{Definition}\thlabel{def1}
	Let $\mathcal{O}$ be an optimization problem with integral cost function $c$. Further, let $\alpha$ is an algorithm which returns a feasible solution $g_\alpha(I)$ at instant $I$ of $\mathcal{O}$ while the optimal value is $g^*_\alpha(I)$. $\alpha$ is called an $\epsilon$-approximation algorithm of $\mathcal{O}$ for some $\epsilon\geq 0$  if and only if 
	$$\rho(I)=\dfrac{|c(g_\alpha(I))-c(g^*_\alpha(I))|}{c(g^*_\alpha(I))}\leq \epsilon ~~\forall I$$
\end{Definition}  
We now design FDOS algorithm such that $\rho(I)$ for the objective function $f(x_{ij})$, denoted as $\rho_f(I)$, is upper bounded by some $\epsilon$ value, which proves FDOS algorithm is an $\epsilon$-approximation algorithm. Let us denote the optimal assignment and the assignment that is obtained from FDOS by $x^*_{ij}$ and $\hat{x}_{ij}$ respectively. Thus, $\rho_f(I)$ is given by eq. (\ref{bound}).
\begin{align}\label{bound}
\rho_f(I)=\dfrac{f(\hat{x}_{ij})-f(x^*_{ij})}{f(x^*_{ij})}
\end{align}
Here, the modulus sign is not required since it is a minimization problem.
The cost function, $f(x_{ij})$, has two components: $f_1(x_{ij})$ and $f_2(x_{ij})$. Let us denote $\rho(I)$ value for $f_1(x_{ij})$ and $f_2(x_{ij})$ as $\rho_{f_1}(I)$ and $\rho_{f_2}(I)$ respectively and they are given by eq. (\ref{fl1}) and eq. (\ref{fl2}).
\begin{align}
\rho_{f_1}(I)=\dfrac{f_1(\hat{x}_{ij})-f_1(x^*_{ij})}{f_1(x^*_{ij})}\label{fl1}\\\label{fl2}
\rho_{f_2}(I)=\dfrac{f_2(x^*_{ij})-f_2(\hat{x}_{ij})}{f_2(x^*_{ij})}
\end{align}
It is evident that if $\rho_{f_1}(I)\leq \epsilon$ and $\rho_{f_2}(I)\leq \epsilon$ then it is sufficient to claim $\rho_f(I)\leq \epsilon$. In this paper, FDOS algorithm will be developed in such a manner that $\rho_{f_1}(I)\leq \epsilon$ and $\rho_{f_2}(I)\leq \epsilon$ which eventually proves that FDOS is a $\epsilon$-approximation algorithm. Toward this target, in FDOS algorithm, we make $f_1(x_{ij})$ upper bounded by some value say $\mathcal{U}_b$ and iteratively change $\mathcal{U}_b$ such that when the algorithm terminate,  $\rho_f(I)\leq \epsilon$. We know, $f_1(x_{ij})=\sum\limits_{j=0}^{M-1}\sum\limits_{i=0}^{N-1}\sum\limits_{k=0}^{N-1}x_{ij}x_{kj}=\sum\limits_{j=0}^{M-1}n_j^2$. Thus, if we upper bound $n_j$ by $\mathcal{U}_b^j$ then $f_1(x_{ij})\leq \sum\limits_{j=0}^{M-1}\mathcal({U}_b^j)^2=\mathcal{U}_b$. Therefore, we include the constraint $n_j=\sum\limits_{i=0}^{N-1}x_{ij}\leq \mathcal{U}_b^j$ $\forall j$ while keeping $f_2(x_{ij})$ in the objective function and then iteratively change the value of $\mathcal{U}_b^j$ such that at the termination $\rho_{f_1}(I)\leq \epsilon$ $\forall I$. We finally prove that the value of $f_2(x_{ij})$, obtained from FDOS algorithm cannot be lower than its optima value of the proposed ILP (refer eq. (\ref{optm})) which proves $\rho_f(I)\leq \epsilon$ or in other word, FDOS is a $\epsilon$-approximation algorithm. Thus, the modified ILP, denoted by $T(\mathcal{U}_b^j)$, turns out to be:
\begin{subequations}\label{optm1}
	\begin{alignat}{2}
	&\!\max_{(i,j)\in\mathcal{A}}        &\qquad&     \sum\limits_{i,j}w_{ij}x_{ij}\label{eq:optProb1}\\
	&\text{subject to} &      & \sum\limits_{j=0}^{M-1}x_{ij}=1~~~~~~\forall i\in\mathbb{N}_s\label{eq:constraint11}\\ & & &\sum\limits_{i=0}^{N-1}x_{ij}\leq \mathcal{U}_b^j~~~~\forall j\in\mathbb{S}_n\label{eq:constraint12}
	\end{alignat}
\end{subequations}
This optimization problem, $T(\mathcal{U}_b^j)$ (refer eq. (\ref{optm1})), is an well known imbalance Transportation Problem \cite{TP} which can be easily converted to balanced Transportation Problem (TP) simply by adding some dummy nodes and by setting $w_{ij}=0~\forall j$ for all dummy nodes \cite{balancing}. Many algorithms are present in the literature for solving the balanced TP in polynomial times complexity (for example \cite{poly}) and any of them can be used for solving $T(\mathcal{U}_b^j)$ (refer eq. (\ref{optm1})). Next, we explain how to change the value of $T(\mathcal{U}_b^j)$ iteratively in FDOS algorithm.

\subsection{FDOS Algorithm}If $N_s$ number of ONUs are assigned in $M_s$ number of slots then there exist at least one slot (say $j$) in which at least  $\big\lceil\frac{N_s}{M_s}\big\rceil$ number of ONUs are assigned. In the first iteration $N_s=N$ and $M_s=M$. In FDOS, we initialize $\mathcal{U}_b^j~(\forall j\in \mathbb{S}_n)$ by $b=\big\lceil\frac{N}{M}\big\rceil$. It can be noted that due to the presence of $\mathcal{A}$, the solution of $T(\mathcal{U}_b^j)$ with $\mathcal{U}_b^j=b=\big\lceil\frac{N}{M}\big\rceil$ may not be feasible. 
If solution of $T(\mathcal{U}_b^j)$ at the initial iteration with $T(\mathcal{U}_b^j)=\lceil\frac{N_s}{M_s}\rceil~\forall j$ is feasible then the assignment, obtained from it, is the final assignment of the FDOS algorithm. 
If the solution is not feasible then set of all slots i.e $\mathbb{S}_n$ is divided into two sets say $\mathbb{L}$ and $\mathbb{O}$ and set of all ONUs i.e $\mathbb{N}_s$ is divided into another two sets say $\mathbb{N}_L$ and $\mathbb{N}_O$ such that in optimal assignment no ONUs of $\mathbb{N}_L$ can be assigned of any of the slots of $\mathbb{O}$ and no ONUs of $\mathbb{N}_O$ can be assigned of any of the slots of $\mathbb{L}$. Generation of theses sets (i.e $\mathbb{L}$, $\mathbb{O}$, $\mathbb{N}_L$, and $\mathbb{N}_O$) will be explained later on this section. Thus, the optimization problem of eq. (\ref{optm}) can be segregated into two sub-problems: one is for assigning ONUs of $\mathbb{N}_L$ into slots of $\mathbb{L}$ and the other is for assigning ONUs of $\mathbb{N}_O$ into slots of $\mathbb{O}$ and the functional value of the actual problem is the summation of the functional values of these two sub-problems. Let us denote the assignment of these two sub-problems as: $xl^*_{ij}$ and $xo^*_{ij}$ respectively and thus, 
\begin{align}\label{sum}
f(x^*_{ij})=f(xl^*_{ij})+f(xo^*_{ij})
\end{align} 

In FDOS, assignment for these two sub-problems are performed separately in the similar way as discussed above. Thus, ONUs of $\mathbb{N}_L$ are assigned into slots of $\mathbb{L}$ by solving $T(\mathcal{U}_b^j)$ with $\mathcal{U}_b^j=\lceil\frac{|\mathbb{N}_L|}{|\mathbb{L}|}\rceil$ while assignment of ONUs of $\mathbb{N}_O$ into slots of $\mathbb{O}$ are performed by solving $T(\mathcal{U}_b^j)$ with $\mathcal{U}_b^j=\lceil\frac{|\mathbb{N}_O|}{|\mathbb{O}|}\rceil$ where $|\mathbb{L}|$, $|\mathbb{O}|$, $|\mathbb{N}_L|$, and $|\mathbb{N}_O|$ denotes cardinality of set $\mathbb{L}$, $\mathbb{O}$, $\mathbb{N}_L$ and $\mathbb{N}_O$ respectively. If the solution of $T(\mathcal{U}_b^j)$ with $\mathcal{U}_b^j=\lceil\frac{|\mathbb{N}_L|}{|\mathbb{L}|}\rceil$ is feasible and the achieved assignment is $\hat{xl}_{ij}$ then the assignment of ONUs ($\in \mathbb{N}_L$) in FDOS is $\hat{xl}_{ij}$. 
If the solution is again infeasible then each of $\mathbb{L}$ and $\mathbb{N}_L$ are  divided into two sets in the similar way as discussed above and the iteration will continue until we get a feasible solution.   
Similar steps will be followed for assigning ONUs of $\mathbb{N}_O$ to slots of $\mathbb{O}$. If the solution is feasible then make the assignment permanent and otherwise,  
 again divide both $\mathbb{O}$ and $\mathbb{N}_O$ into two sets as discussed. This process will continue until a feasible solution is achieved for all ONUs and in later part we prove that the algorithm always converges (refer Section \ref{sssec:comple}). All iterations of FDOS algorithm are provided in Algorithm \ref{algo2}.  
 \begin{algorithm}[h]
 	\SetAlgoLined
 	\KwIn{$\mathcal{N}(\subseteq \mathbb{N}_s)$- Set of all ONUs for whom assignment is performed, $\mathcal{S}(\subseteq \mathbb{S}_n)$- Set of slots at which assignment of $\mathcal{N}$ is performed.}
 	\textbf{Initialization:} $\mathcal{N}=\mathbb{N}_s$, $\mathcal{S}=\mathbb{S}_n$\;
 	Solve $T\big(\lceil\frac{|\mathcal{N}|}{|\mathcal{S}|}\rceil\big)$ (refer eq. (\ref{optm1}))\;
 	\eIf{solution is feasible~\tcp{stopping criteria}}{Make the assignment of all ONUs ($\in \mathcal{N}$) parmanent\;
 		return\;}
 	{	
 		Find sets $\mathbb{L}$, $\mathbb{O}$, $\mathbb{N}_L$ and $\mathbb{N}_O$ using Algorithm 2\;	
 		call FDOS($\mathbb{N}_L$,$\mathbb{L}$)\;~\tcp{recursive call for slots ($\in \mathbb{L}$)}
 		call FDOS($\mathbb{N}_O$,$\mathbb{O}$)\;~\tcp{recursive call for slots ($\in \mathbb{O}$)}
 	}
 	\caption{Proposed FDOS($\mathcal{N}$,$\mathcal{S}$) algorithm}
 	 	 \label{algo2}
 \end{algorithm}
\subsection{Properties of FDOS algorithm}
Here, we prove that FDOS is an  1-approximation algorithm. In FDOS, if $N$ number of ONUs are assigned in $M$ number of slots then the first step is to solve balanced TP problem, $T(\mathcal{U}_b^j)$, with $\mathcal{U}_b^j=\Big\lceil\dfrac{N}{M}\Big\rceil=b$. If the solution is feasible then  in \thref{cl1}, we prove $\rho_f(I)$ is upper bounded. 
\begin{claim}\thlabel{cl1}
	If solution of $T(\mathcal{U}_b^j)$ with $\mathcal{U}_b^j=b~\forall j$ is feasible then $\rho_{f}(I)\leq 1~\forall I$.
\end{claim}
\begin{proof}
 \thref{cl1} can be easily proved by using \thref{cl2} and \thref{cl3}.
\end{proof}
\begin{claim}\thlabel{cl2}
	If solution of $T(\mathcal{U}_b^j)$ with $\mathcal{U}_b^j=\lceil\frac{N_s}{M_s}\rceil= b~\forall j$ is feasible then $\rho_{f_1}(I)\leq 1~\forall I$.
\end{claim}
\begin{proof}
	In \thref{l2}, we prove;
	\begin{align}\label{eneq1}
	f_1(x^*_{ij})\geq M_sb^2-(2b-1)(M_sb-N_s)
	\end{align}
	Further, in \thref{l3}, we prove, $$f_1(\hat{x}_{ij})\leq\Big\lfloor\frac{N_s}{b}\Big\rfloor b^2+\Big(N_s-\Big\lfloor\frac{N_s}{b}\big\rfloor b\Big)^2$$
	It is evident that $\big(N_s-\lfloor\frac{N_s}{b}\rfloor b\big)\leq b$. Thus, 
	\begin{align}\label{eneq2}
	f_1(\hat{x}_{ij})\leq\Big\lfloor\frac{N_s}{b}\Big\rfloor b^2+b^2\leq b(N_s+b)
	\end{align}
	By using the inequalities of eq. (\ref{eneq1}) and eq. (\ref{eneq2}) in eq. (\ref{bound}), we get:
	\begin{align}\label{rho}
	\rho_{f_1}(I)\leq \dfrac{M_sb(b-1)-N_s(b-1)+b^2}{M_sb^2-(2b-1)(M_sb-N_s)} ~~~~\forall I
	\end{align}
	Clearly, right hand side of eq. (\ref{rho}) is a decreasing function of $N_s$ as both $b$ and $M$ are positive. Further, since $b=\lceil\frac{N_s}{M_s}\rceil$, $N_s\geq (b-1)M_s+1$. Substituting $N_s=(b-1)M_s+1$ in eq. (\ref{rho}), we get
	\begin{align}
	\rho_{f_1}(I)\leq \dfrac{(M_s-1)(b-1)+b^2}{(M_s-1)(b-1)^2+b^2}~~~~\forall I
	\end{align}
	It can be easily proved that $\rho_{f_1}(I)\leq 1$ for $b\geq 2$.
	For $b=1$, \thref{l1} proves that in optimal assignment, every slot is occupied by at-most one ONU. Therefore, in optimal assignment (i.e ${x}^*_{ij}$), exactly $N_s$ number of slots are occupied by one ONU and rest of the slots remain unoccupied which is the case for $\hat{x}_{ij}$ as well. Thus, $f_1(\hat{x}_{ij})=f_1({x}^*_{ij})$ which proves \thref{cl2}. 
\end{proof}
\begin{claim}\thlabel{cl3}
		If solution of $T(\mathcal{U}_b^j)$ with $\mathcal{U}_b^j=b~\forall j$ is feasible then $f_2(\hat{x}_{ij})-f_2({x}^*_{ij})\geq 0$ where ${x}^*_{ij}$ is the optimal assignment of eq. (\ref{optm}) and $\hat{x}_{ij}$ is the assignment obtained by FDOS algorithm.
\end{claim}
\begin{proof}
	In $T(\mathcal{U}_b^j)$, only one extra constraint is added as compared to the actual optimization problem (refer eq. (\ref{optm})) and the constraint is $\sum\limits_{i=0}^{N_s-1}x_{ij}\leq b$. \thref{l1} proves that in optimal assignment, every slot satisfies this constraint which further proves 
	\thref{cl3}.
\end{proof} 
If the solution is not feasible then set of all slots are divided into two mutually exclusive non-empty sub-sets (non-emptiness will be proved in Section \ref{sssec:comple}) and the process continues until all ONUs get feasible assignment. Thus, the FDOS algorithm divides both $\mathbb{N}_s$ and $\mathbb{S}_n$ into some $K$ number of sub-sets, denoted by $S_p$ and $N_p$ respectively where $p\in \{1,2,\dots, K\}$, such that in optimal assignment, ONUs of $N_p$ can only be assigned to slots of $S_p$. In FDOS, assignment of $N_p$ in $S_p$ is performed by solving $T\Big(\dfrac{|N_p|}{|S_p|}\Big)$ (refer eq. (\ref{optm1})). If the optimal assignment and the assignment achieved by the FDOS algorithm of $N_p$ in $S_p$ is denoted by $xp^*_{ij}$ and $\hat{xp}_{ij}$ respectively then \thref{cl1} proves that 
\begin{align}\label{ene}
\dfrac{f(\hat{xp}_{ij})-f(xp^*_{ij})}{f(xp^*_{ij})}\leq 1~~~~ \forall 
p\in \{1,2,\dots,K\}
\end{align}
By using the inequality of eq. (\ref{ene}), in \thref{cl4}, we now prove that  FDOS is an $1$-approximation algorithm (i.e. $\rho_{f}(I)\leq 1~\forall I$). 
\begin{claim}\thlabel{cl4}
	FDOS ia an $1$-approximation algorithm.
\end{claim}
\begin{proof}
	The FDOS algorithm divides both $\mathbb{N}_s$ and $\mathbb{S}_n$ into exactly $K$ number of sub-sets say $N_p$ and $S_p$ $\forall p\in\{1,2,\dots,K\}$ and the assignment of $N_p$ are performed only in $S_p$. Clearly, 
	\begin{align}\label{c1}
	f(\hat{x}_{ij})=\sum\limits_{i=1}^Kf(\hat{xp}_{ij})
	\end{align}
	Further, from eq. (\ref{sum}), we get 
	\begin{align}\label{c2}
	f({x}^*_{ij})=\sum\limits_{i=1}^Kf({xp}^*_{ij})
	\end{align}
	Applying eq. (\ref{c1}) and eq. (\ref{c2}) in eq. (\ref{bound}), we get:
	\begin{align}\label{c3}
	\rho_{f}(I)=\dfrac{\sum\limits_{i=1}^Kf(\hat{xp}_{ij})-\sum\limits_{i=1}^Kf({xp}^*_{ij})}{\sum\limits_{i=1}^Kf({xp}^*_{ij})}
	\end{align}
	By applying the inequality of eq. (\ref{ene}) in eq. (\ref{c3}), it can be easily proved that $\rho_{f}(I)\leq 1$ or in other word, FDOS is an $1$-approximation algorithm.
\end{proof}
Next, we explain the algorithm for generating the sets $\mathbb{L}$, $\mathbb{O}$, $\mathbb{N}_L$, and $\mathbb{N}_O$ and then we prove that in optimal assignment any ONU of $\mathbb{N}_L$ can not be assigned to any slot of $\mathbb{O}$ and any ONU of $\mathbb{N}_O$ can not be assigned to any slot of $\mathbb{L}$.

\subsection{{Generation of sets  ${\mathbb{L}}$, ${\mathbb{O}}$, ${\mathbb{N}_L}$, and ${\mathbb{N}_O}$ and their properties}} 
As we explained above, if the solution of eq. (\ref{optm1}) is infeasible then  $\mathbb{S}_n$ is divided into $\mathbb{L}$ and $\mathbb{O}$ (i.e $\mathbb{S}_n=\mathbb{L}\cup\mathbb{O}$) and $\mathbb{N}_s$ is divided into $\mathbb{N}_L$ and $\mathbb{N}_O$ (i.e $\mathbb{N}_s=\mathbb{N}_L\cup\mathbb{N}_O$). Since, solution of eq. (\ref{optm1}) is infeasible, our next step is to find feasible assignment for maximum number of ONUs. In order to do so, we modify eq. (\ref{optm1}) in the following ways. Firstly, we remove the arc set constraint (i.e $(i,j)\in \mathcal{A}$) which allows an assignment for all ONUs. Definitely, some assignments must be infeasible (i.e doesn't belong to $\mathcal{A}$). Further, we modify the cost factor, $w_{ij}$ by $w'_{ij}$ where  $w'_{ij}$ is given by eq. (\ref{weight}).
\begin{align}\label{weight}
w'_{ij}=\begin{cases}
w_{ij}~~~~&if~(i,j)\in\mathcal{A}\\
-H &otherwise 
\end{cases}
\end{align} 
It is quite evident that if $H>\sum\limits_{(i,j)\in\mathcal{A}} w_{ij}$ then we always get an feasible assignment for maximum number of ONUs. Thus, the modified version of eq. (\ref{optm1}) that we solve for getting assignment, is given by eq. (\ref{optm2}).
\begin{subequations}\label{optm2}
	\begin{alignat}{2}
	&\!\max_{i\in\mathbb{N}_s, j\in\mathbb{S}_n}        &\qquad&     \sum\limits_{i,j}w'_{ij}x_{ij}\label{eq:optProb2}\\
	&\text{subject to} &      & \sum\limits_{j=0}^{M-1}x_{ij}=1~~~~~~\forall i\in\mathbb{N}_s\label{eq:constraint21}\\ & & &\sum\limits_{i=0}^{N-1}x_{ij}\leq \mathcal{U}_b^j~~~~\forall j\in\mathbb{S}_n\label{eq:constraint22}
	\end{alignat}
\end{subequations}
Let, the us denote the assignment, achieved by solving eq. (\ref{optm2}), as $\bar{xm}_{ij}$. Clearly, $\bar{xm}_{ij}$ includes some assignments that $\notin \mathcal{A}$. Let, the set of all these ONUs as $\mathbb{U}_s$ and it is given by 
\begin{align}
\mathbb{U}_s=\{i|\bar{xm}_{ij}=1, (i,j)\notin \mathcal{A}\}
\end{align}
Let us form an assignment, say $\bar{x}_{ij}$, by removing assignments for all ONUs ($\in\mathbb{U}_s$) from $\bar{xm}_{ij}$ (i.e make $\bar{x}_{ij}=0~ \forall i\in \mathbb{U}_s, j\in\mathbb{S}_n$). Clearly, all assignments of $\bar{x}_{ij}$ are feasible and $\bar{x}_{ij}$ includes feasible assignments for maximum possible number of ONUs. Let $\mathbb{M}_j$ denotes the set of ONUs that are assigned to slot $j$ and is given by eq. (\ref{mj}).
\begin{align}\label{mj}
\mathbb{M}_j=\{i|\bar{x}_{ij}=1~~\forall i\in \mathbb{N}_s\}
\end{align}
\thref{cl5} proves that any ONU of $\mathbb{U}_s$ cannot be assigned to any slot for which $|\mathbb{M}_j|<\mathcal{U}_b^j$. Let us include all these slots into set $\mathbb{L}$. Thus, $\mathbb{L}$ is initialized by eq. (\ref{Lset}).
\begin{align}\label{Lset}
\mathbb{L}=\{j||\mathbb{M}_j|<\mathcal{U}_b^j~~\forall j\in \mathbb{M}_s\}
\end{align} 
\begin{claim}\thlabel{cl5}
	There doesn't exists any $i\in \mathbb{U}_s$ for which $(i,j)\in\mathcal{A}$ $\forall j$ $\ni |\mathbb{M}_j|<\mathcal{U}_b^j$.
\end{claim} 
\begin{proof}
	We prove this claim by the method of contradiction. Suppose there exists an $(i,j)$ pair such that $i\in \mathbb{U}_s$, $|\mathbb{M}_j|<\mathcal{U}_b^j$ and $(i,j)\in \mathcal{A}$. Let in $\bar{xm}_{ij}$, ONU $i$ is assigned to slot $k$. Now consider another assignment, say $xm'_{ij}$, which is formed by assigning ONU $i$ to slot $j$. Clearly, $\sum\limits_{i,j}w'_{ij}\bar{xm}_{ij}<\sum\limits_{i,j}w'_{ij}xm'_{ij}$ which contradicts that $\bar{xm}_{ij}$ is the optimal solution of eq. (\ref{optm2}).
\end{proof} 
Further, \thref{cl6} proves that if there exists a slot in $\mathbb{S}_n\setminus\mathbb{L}$, say $j'$, such that an $(i,j)$ pair is present in $\mathcal{A}$ where $i\in \mathbb{M}_{j'}$ and $j\in\mathbb{L}$ then any ONU of $\mathbb{U}_s$ cannot be assigned to $j'$. Thus, if $j'$ is included in set $\mathbb{L}$ then also any ONU of $\mathbb{U}_s$ cannot be assigned to any slots of $\mathbb{L}$. 
All such slots are included to the $\mathbb{L}$ set and it is performed in the following ways. We go through all ONUs that are assigned to a slot of $\mathbb{S}_n\setminus\mathbb{L}$ in $\bar{x}_{ij}$ and check whether there exists an ONU, say $i$, which can also be assigned to a slot, say $j\in\mathbb{L}$ or in other word $(i,j)\in \mathcal{A}$. If such an ONU exists and it is assigned to slot $j'$ then include $j'$ in $\mathbb{L}$ (i.e $\mathbb{L}\cup j'$). This process will continue until there is no such ONU is present. All steps of generation of set $\mathbb{L}$ is shown in Algorithm 3. All other slots are included to the other set i.e $\mathbb{O}$. Thus, $\mathbb{O}$ is given by eq. (\ref{oset}).
\begin{align}\label{oset}
\mathbb{O}=\mathbb{S}_n\setminus \mathbb{L}
\end{align} 
\begin{claim}\thlabel{cl6}
	If there exists an ONU, say $i$, which is assigned to slot $j'\in \mathbb{S}_n\setminus\mathbb{L}$ in $\bar{x}_{ij}$ and there exists a slot, say $j\in\mathbb{L}$, such that $(i,j)\in \mathcal{A}$ then $(k,j')\notin \mathcal{A}$ $\forall k\in \mathbb{U}_s$.
\end{claim}
\begin{proof}
	We prove this claim by the method of contradiction. Suppose, there exists an ONU, say $k'\in \mathbb{U}_s$ which can be assigned to slot $j'$. Find another assignment, say $xm'_{ij}$, from $xm_{ij}$ by assigning ONU $k'$ to slot $j'$ and ONU $i$ to slot $j$. Clearly,  $xm'_{ij}$ is a feasible solution of eq. (\ref{optm2}) and $\sum\limits_{i,j}w'_{ij}\bar{xm}_{ij}<\sum\limits_{i,j}w'_{ij}xm'_{ij}$. This contradict that $xm_{ij}$ is the optimal assignment of eq. (\ref{optm2}).
\end{proof}
Let, $\mathbb{N}_L$ and $\mathbb{N}'_O$ denotes the set of ONUs that are assigned to sets $\mathbb{L}$ and $\mathbb{O}$ respectively. Further, let us generate a set $\mathbb{N}_O$ by $\mathbb{N}_O=\mathbb{N}'_O\cup \mathbb{U}_s$. All steps of generation of sets $\mathbb{L}$, $\mathbb{O}$, $\mathbb{N}_L$ and $\mathbb{N}_O$ are shown in Algorithm 2. In \thref{cl7} and \thref{cl8}, we prove that any ONU of $\mathbb{N}_L$ and $\mathbb{N}_O$ cannot be assigned of any slots of $\mathbb{O}$ and $\mathbb{L}$ respectively. 
\begin{claim}\thlabel{cl7}
	There doesn't exists any $(i,j)$ pair $\ni$ $(i,j)\in \mathcal{A}$ $\forall i\in \mathbb{N}_O$ and $j\in \mathbb{L}$.
\end{claim}
\begin{proof}
 \thref{cl5} and \thref{cl6} jointly proves that there doesn't exists any $(i,j)$ pair $\ni$ $(i,j)\in \mathcal{A}$ $\forall i\in \mathbb{U}_s$ and $j\in \mathbb{L}$. Since, $\mathbb{N}_O=\mathbb{N}'_O\cup \mathbb{U}_s$, now, we have to prove that there doesn't exists any $(i,j)\in \mathcal{A}$ $\ni i\in \mathbb{N}'_O$ and $j\in \mathbb{L}$. 
 We prove this by the method of contradiction. Suppose there exists such an pair, say $(i,j)\in \mathcal{A}$ where $i\in \mathbb{N}'_O$ and $j\in \mathbb{L}$. Further, let in $\bar{x}_{ij}$, ONU $i$ is assigned to slot $k$ which implies $k\in \mathbb{O}$. Since, $(i,j)\in \mathcal{A}$ and $i$ in assigned to a slot of $\mathbb{O}$, $k$ should be included in set $\mathbb{L}$ (refer Algorithm 2) which contradicts that $k\in \mathbb{O}$ 
 
\end{proof}
\begin{claim}\thlabel{cl8}
	In optimal assignment of eq. (\ref{optm}), no ONU of $\mathbb{N}_L$ can be assigned to any slot of $\mathbb{O}$.
\end{claim}
\begin{proof}
We prove this claim  by the method of contradiction. Suppose, $x^*_{ij}$ is the optimal assignment where there exists an ONU, say $i\in\mathbb{N}_L$, which is assigned to slot $j\in \mathbb{O}$. If we can find an assignment, say $\hat{x}_{ij}$ $\ni$ $f(\hat{x}_{ij})<f(x^*_{ij})$ then it contradicts the optimality of $x^*_{ij}$.
 As $W>\max\limits_{(i,j)\in\mathcal{A}}f_2(x_{ij})$, if we can prove $f_1(\hat{x}_{ij})\leq f_1(x^*_{ij})-1$ then it is sufficient to prove the claim.   
  Since, $i\in\mathbb{N}_L$, there must exists a slot, say $k\in\mathbb{L}$ $\ni$ $(i,k)\in \mathcal{A}$. Let, $x'_{ij}$ is an assignment which is achieved by assigning ONU $i$ in slot $k$ instead of slot $j$. Clearly, $f_1({x}^*_{ij})-f_1(x'_{ij})=2(n_j^*-n_k^*-1)$. If $n_j^*-n_k^*\geq 2$ then $f_1({x}^*_{ij})-f_1(x'_{ij})\geq 2$ and hence, $\hat{x}_{ij}$ can be chosen as $x'_{ij}$. We now find $\hat{x}_{ij}$ for other two cases: $n_j^*-n_k^*=1$, $n_j^*-n_k^*=0$.
  
 \subsubsection*{Case $n_j^*-n_k^*=1$ }
  In this case, $f_1({x}^*_{ij})=f_1(x'_{ij})$.
  Now, \thref{l4} and \thref{l5} prove that in $x^*_{ij}$, $n_{j}^*\geq b$ and $n_{k}^*\leq b$ respectively. Thus, in this case, there exist exactly two pairs of values of $n_j^*$, $n_k^*$ are feasible: $n_j^*=b$, $n_k^*=b-1$ and $n_j^*=b+1$, $n_k^*=b$. For the first case, in $x'_{ij}$, there exists an slot $j\in \mathbb{O}$ which is assigned with lower than $b$ number of ONUs. \thref{l4} proves that we can always find another assignment $x''_{ij}$ $\ni$ $f_1(x''_{ij})\leq f_1(x'_{ij})-2$. Similarly, using \thref{l5}, we can prove that for the second case also there exists such an $x''_{ij}$. So, in this case, we can choose $\hat{x}_{ij}$ as ${x}''_{ij}$. 
  \subsubsection*{Case $n_j^*-n_k^*=0$ }
   In this case, $f_1({x}^*_{ij})=f_1(x'_{ij})-2$ and there is only one feasible pair of values of $n_j^*$, $n_k^*$ is present: $n_j^*=b$, $n_k^*=b$. Since, in $x'_{ij}$, there exists an slot $j\in \mathbb{O}$ which is assigned with lower than $b$ number of ONUs \thref{l4} proves that we can always find another assignment $x''_{ij}$ $\ni$ $f_1(x''_{ij})\leq f_1(x'_{ij})-2$. Further, in $x''_{ij}$, there exists an slot $k\in \mathbb{L}$ which is assigned with more than $b$ number of ONUs and hence, \thref{l5} proves that we can always find another assignment $x'''_{ij}$ $\ni$ $f_1(x'''_{ij})\leq f_1(x''_{ij})-2$. Therefore, $f_1({x}'''_{ij})\leq f_1(x^*_{ij})-2$ and hence,  we can choose $\hat{x}_{ij}$ as ${x}'''_{ij}$.  
\end{proof}
\begin{algorithm}[t]
	\SetAlgoLined
	\KwIn{$\mathcal{N}$, $\mathcal{S}$}
	\textbf{Initialization:} $\mathbb{L}=\phi$\;
	
		\ForEach{$j\in \mathcal{S}$}
		{
			\If{$|\mathbb{M}_j|<\Big\lceil\dfrac{|\mathcal{N}|}{|\mathcal{S}|}\Big\rceil$}
			{ $\mathbb{L}=\mathbb{L}\cup j$\;}
		}
		%
		flag$=0$\;
		\While{flag$=0$}
		{
			flag$=1$\;
			\ForEach{$i\in \bigcup\limits_{j\in\mathcal{S}\setminus\mathbb{L}}\mathbb{M}_j$}
			{
				\ForEach{$ j'\in \mathbb{L}$}
				{
					\If{$(i,j')\in \mathcal{A}$}
					{
						 $\mathbb{L}=\mathbb{L}\cup j$, flag$=0$\;

					}
				}	
			}
		}
	 $\mathbb{O}=\mathcal{S}\setminus \mathbb{L}$, $\mathbb{N}_L=\bigcup\limits_{j\in\mathbb{L}}\mathbb{M}_j$, $\mathbb{N}_O=\bigcup\limits_{j\in\mathbb{O}}\mathbb{M}_j$;\ 		
	\caption{Generation of sets $\mathbb{L}$, $\mathbb{O}$, $\mathbb{N}_L$, and $\mathbb{N}_O$}
\end{algorithm} 
\subsection{Convergence Analysis}\label{sssec:comple}
Here, we prove that the FDOS algorithm converges. The FDOS algorithm assigns ONUs of set $\mathbb{N}_s$ in slots of set $\mathbb{S}_n$ and the first step is to solve the unbalance Transportation Problem (TP). If the solution is feasible then the algorithm terminates and otherwise, set $\mathbb{S}_n$ is divided into exactly two mutually exclusive sets $\mathbb{L}$ and $\mathbb{O}$. 
For these two sets $\mathbb{L}$ and $\mathbb{O}$ the same steps are followed. 
Thus, proving $\mathbb{L}\neq \emptyset$ and $\mathbb{O}\neq \emptyset$ are sufficient to prove the convergence of FDOS algorithm.  Since, $|\mathbb{N}_s|=N$ number of ONUs are assigned in $|\mathbb{S}_n|=M$ number of slots and the sets $\mathbb{L}$ and $\mathbb{O}$ are generated if at least one ONU remains unassigned (i.e. $\mathbb{U}_s\neq \Phi$), there must exist at least one slot which is assigned with less that $\Big\lceil\dfrac{N}{M}\Big\rceil$ number of ONUs (using Prison hole principal). This slot will be part of $\mathbb{L}$ which proves $\mathbb{L}\neq \emptyset$. Now, we prove that if there exist at least one $j$ $\ni$ $(i,j)\in \mathcal{A}$ $\forall i$ then $\mathbb{O}\neq \emptyset$. Suppose, the assignment of ONU $i$ is infeasible and $(i,k)\in \mathcal{A}$ then \thref{cl6} proves that $k\notin\mathbb{L}$ and hence, $k\in\mathbb{O}$ which proves $\mathbb{O}\neq \emptyset$  
\begin{figure*}[t]
	\centering
	\includegraphics[scale=.55]{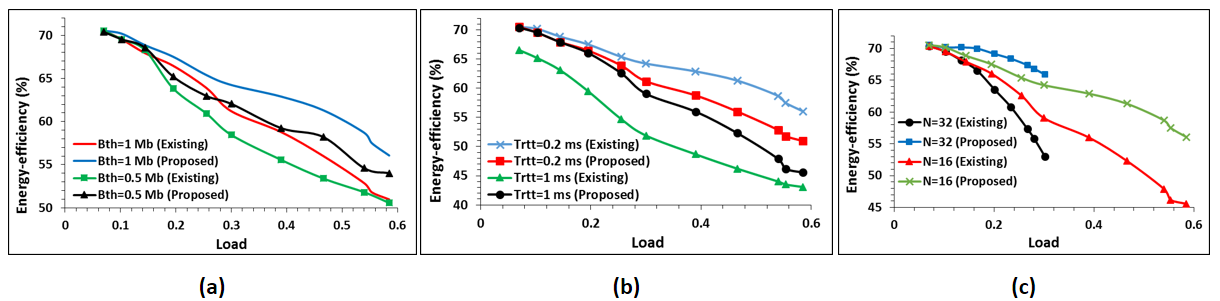}
	\caption{Energy-efficiency for different values of (a) $B_{th}^i$ (b) $T_{rtt}^i$ (c) $N$. Bth- $B_{th}^i$, Trtt- $T_{rtt}^i$, existing- results obtained from OSMP-EO protocol}
	\label{fig:energy}
\end{figure*}
\subsection{Complexity Analysis} 
Here, we analyze the asymptotic complexity of the FDOS algorithm. Let us consider that $N$ number of ONUs are assigned in $M$ number of slots. The first step is to assign ONUs by solving the Transportation Problem (TP). Let us denote the time complexity of solving TP as $C_T(m,n)$. In \cite{poly}, the authors has proposed an algorithm for solving TP with a complexity of $\mathcal{O}\big(M\log M(|\mathcal{A|}+N\log N)\big)$. If the solution is infeasible then slots and ONUs are divided in two sets $\mathbb{L}$, $\mathbb{O}$ and $\mathbb{N}_L$, $\mathbb{N}_O$ respectively using Algorithm 3. Generation of set $\mathbb{M}_j$ requires going through all ONUs requiring a complexity of $\mathcal{O}(N)$. The next step is to include all slots that are assigned with lower than $\Big\lceil\dfrac{N}{M}\Big\rceil$ number of ONUs to set $\mathbb{L}$ which requires going through all slots and hence, complexity is $\mathcal{O}(M)$. Next step is to go through all ONUs that are not assigned to any of the slots of $\in\mathbb{L}$ and for every ONU, check whether there exists any $k$ $\ni$ $(i,k)\in \mathcal(A)$. This process requires going through every slots for every ONUs requiring a complexity of $\mathcal{O}(MN)$ (refer inner two for loops in Algorithm 3). The  Algorithm enters into the while loop of Algorithm 3 at max $\min(M,N)$ times as in every iteration, at least one ONU and one slot get removed. Thus, generation of $\mathbb{L}$ amd $\mathbb{O}$ requires a time complexity of $\mathcal{O}\big(C_T(m,n)+\min(M,N)MN\big)$ which we denote as $C(M,N)$. The same step will be followed for both of these two sets. Thus, if $|\mathbb{L}|=M_1$ and $|\mathbb{N}_L|=N_1$ then $|\mathbb{O}|=M-M_1$ and $|\mathbb{N}_O|=N-N_1$ and hence, the time complexity is: $C(M_1,N_1)+C(M-M_1,N-N_1)$. It is quite evident that $C(M_1,N_1)+C(M-M_1,N-N_1)\leq C(M,N)$. Thus, in every stage, the complexity is upper bounded by: $C(M,N)$. As we have proved in Section \ref{sssec:comple} that every stage the cardinality of every sets (i.e all $\mathbb{L}$, $\mathbb{O}$, $\mathbb{N}_L$ and $\mathbb{N}_O$) are reduced at least by one as compared to their parent sets. Thus, the number of stages are upper bounded by $\min(M,N)$. Thus, the complexity of the FDOS algorithm is given by: $\min(M,N)C(M,N)$. 
\section{Results and discussions}
In this paper, we propose a protocol that is followed by the OLT for scheduling the wake-up message while the ONUs follows our previously proposed OSMP-EO protocol []. The wake-up scheduling allows the OLT to fairly distribute the active ONUs among cycles which provides improvement in energy-efficiency (refer Section \ref{intro}). In order to quantify this improvement, we compare energy-efficiency figures of the  OSMP-EO protocol with our proposed mechanism (i.e. OSMP-EO protocol with wake-up scheduling). We also compare the the average delay figure of both of them.  All results are generated from simulations, performed in OMNET++ for a network runtime of $50s$  and they are plotted with $95\%$  confidence interval.  Since, the OSMP-EO protocol is designed only for the US traffic, for fair comparison, in simulations, we consider only the US traffic.
The link rate of the feeder fiber and the maximum traffic arrival rate at each ONU are assumed to be $1$ Gbps and $100$ Mbps respectively \cite{chayan}. All packets are  Ethernet packets of size $1500$ Bytes. The traffic arrivals are considered to be self-similar which is generated by aggregating $16$ ON-OFF Pareto sources. The Hurst parameter ($H$) of the self similar traffic are considered as $0.8$. The buffer size of each ONUs are $1.2$ Mb ($100$ packets). Sleep-to-wake-up time of $ds$, $fs$, and $dz$ are considered to be $5.125~$ms, $125~\mu s$ and $1~\mu s$ respectively \cite{chayan}. Whereas power consumption of $ds$, $fs$, $dz$ and $on$ are $0.75$ W, $1.28$ W, $2.39$ W and $3.984$ W respectively \cite{chayan}. We consider the de-registration time as $50$ ms \cite{ieee2010ieee}.
\subsection{Energy-efficiency}
Here, we compare the energy-efficiency figures of our proposed mechanism with the same of the OSMP-EO protocol for different values of $B_{th}^i$, $T_{rtt}^i$, and $N$. 
\subsubsection{Effect of $B_{th}^i$} \label{sssec:bth}
Energy-efficiency figures of the OSMP-EO protocol and our proposed mechanism for $B_{th}^i=1$ Mb and $0.5$ Mb are plotted in Fig. \ref{fig:energy}(a). Here, we consider $T_{rtt}^i$ ($\forall i$) and $N$ as $0.2$ ms and $16$ respectively. It can be observed from the figure that a significant improvement in energy-efficiency can be achieved by fairly distributing ONUs among cycles for both  $B_{th}^i=1$ Mb and $0.5$ Mb. Further, an increment of $B_{th}^i$ improves the energy-efficiency figure as shown in Fig. \ref{fig:energy}(a). This is because of the following facts. The buffer fill-up time increases with an increase in $B_{th}^i$. As a result, ONUs wakes-up from sleep mode less frequently. We know that whenever an ONU wakes-up from a sleep mode $S_m$, on an average $T_{sw}^{S_m}+1.5T_{avg}$ ($T_{avg}$ denotes the average cycle time) duration is wasted (refer Section \ref{sssec:osmp}) when the ONUs are active but no US data is transmitted. Reduction of $B_{th}^i$ increases this waste and hence, the achieved energy-efficiency diminishes with a decrement of $B_{th}^i$. 
Further, an increment of traffic load also decreases the buffer fill-up time and hence, the number of slots, at which the active ONUs can be assigned, reduces. Thus, the possibility that in a slot, more number of ONUs wakes-up as compared to the others is quite high (i.e unfair assignment) and hence, the effect of fair distribution is more at higher load. Therefore, the improvement of energy-efficiency as compared to the OSMP-EO protocol enhances with the increment of traffic load as shown in Fig. \ref{fig:energy}(a). 

\subsubsection{Effect of $T_{rtt}^i$} \label{sssec:trtt}
Energy-efficiency figures of the OSMP-EO protocol and our proposed mechanism for $T_{rtt}^i=0.2$ ms (reach is $20$ km) and $1$ ms (reach is $100$ km) are plotted in Fig. \ref{fig:energy}(b). The considered values of $N$ and $B_{th}^i$ are $16$ and $1$ Mb respectively. It can be observed from Fig. \ref{fig:energy}(b) that the improvement that is achieved by employing our proposed mechanism as compared to the OSMP-EO protocol enhances with the increment of $T_{rtt}^i$ especially at low and medium load. This is because of the following facts. An increment of $T_{rtt}^i$ increases the duration of minimum value of polling cycles. At low and medium load the buffer fill-up time being quite high, in a cycle, only few ONUs are active. Thus, in OSMP-EO, a certain portion of a cycle is wasted when no US data is transmitted. This increases the duration of active periods of ONUs. However, in our proposed mechanism, the entire cycle is distributed among the active ONUs and as a result, this waste gets reduced. Further, we know every time an ONU wakes-up from sleep mode, on an average, $T_{sw}^{S_m}+1.5T_{avg}$ duration is wasted. An increase in $T_{rtt}^i$ increases $T_{avg}$ and hence, this waste. Thus, the energy-efficiency figure of our proposed mechanism enhances with a decrement of $T_{rtt}^i$ as shown in Fig. \ref{fig:energy}(b).  
\subsubsection{Effect of $N$}  \label{sssec:N}  
 Energy-efficiency figures of the OSMP-EO protocol and our proposed mechanism for $N=16$ and $32$ are plotted in Fig. \ref{fig:energy}(b). The considered values of  $B_{th}^i$ and $T_{rtt}^i$ ($\forall i$) are  $1$ Mb and $0.2$ ms respectively. An increment of $N$ increases the duration of polling cycles and hence, lesser number of cycles are allocated with more number of ONUs. As a results, the allocation is  mostly unfair. Therefore, the possibility of fair allocation of the active ONUs on cycles is more for a higher value $N$. Thus, the improvement of energy-efficiency as compared to OSMP-EO enhances with an increase in $N$ as shown in Fig. \ref{fig:energy}(c).     
 \begin{figure*}[h]
 	\centering
 	\includegraphics[scale=.55]{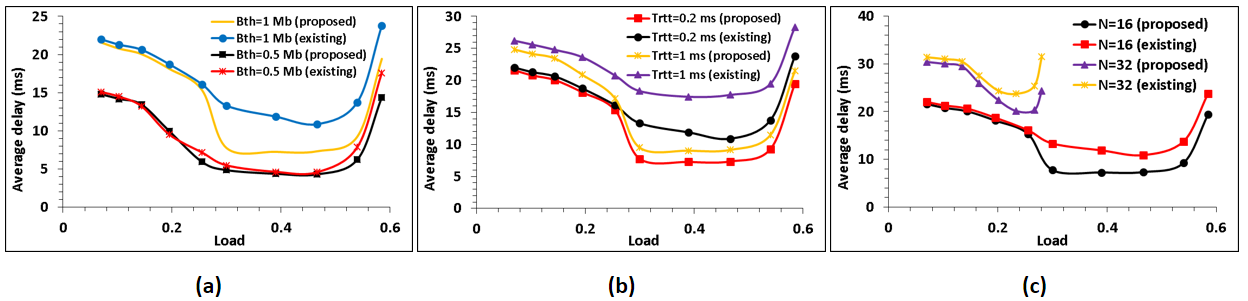}
 	\caption{Average delay for different values of (a) $B_{th}^i$ (b) $T_{rtt}^i$ (c) $N$. Bth- $B_{th}^i$, Trtt- $T_{rtt}^i$, existing- results obtained from OSMP-EO protocol}
 	\label{fig:delay}
 \end{figure*}
 \subsection{Average Delay} 
 Here, we observe the effect of $B_{th}^i$, $T_{rtt}^i$, and $N$ on average delay figures of our proposed mechanism. We also compare them with the average delay figures of the OSMP-EO protocol. In order to do so, in Fig. \ref{fig:delay}(a), Fig. \ref{fig:delay}(b), and Fig. \ref{fig:delay}(c), we have plotted the average delay figures of our proposed mechanism and the OSMP-EO protocol for different values of $B_{th}^i$ ($1$ Mb and $0.5$ Mb), $T_{rtt}^i$ ($0.2$ ms and $1$ ms), and $N$ ($16$ and $32$) respectively. All other considered parameters of Fig. \ref{fig:delay}(a), Fig. \ref{fig:delay}(b), and Fig. \ref{fig:delay}(c) are same as Fig. \ref{fig:energy}(a), Fig. \ref{fig:energy}(b), and Fig. \ref{fig:energy}(c) respectively. It can be observe from Fig. \ref{fig:delay} that at low load, the delay figures reduce with an increment of traffic load and after a entertain critical load, exactly the opposite trend can be seen. This is due to the following facts. The average delay figures have two components: (i) delay due to sleep duration, (ii) queuing delay during the active periods. At low load, the buffer fill-up times are quite high causing a higher value of sleep duration while the queuing delay is very small. Therefore, at low load, the average delay figures are mostly determined by sleep durations. Sleep durations and hence, average delay decreases with the increment of traffic load. Increment of traffic load decreases the sleep duration. Further, it increases the queuing delay as well. As a consequence, after a certain critical load, the queuing delay dominates the sleep duration. Hence, the average delay figures start increasing with an increment of traffic load after that critical value of traffic load. Another interesting observation is that the average delay figures of our proposed mechanism are much lower as compare to the same of the OSMP-EO protocol. This is because, in OSMP-EO, ONUs try to wake-up from sleep mode and initiate the US transmission just before their buffers fill up while in our proposed mechanism, the OLT forces ONUs to wake-up little early in order to fairly distribute them among cycles.   As discussed in Section \ref{sssec:bth}, increment of $B_{th}^i$ enhances the sleep duration and hence, the average delay as shown in Fig. \ref{fig:delay}(a). Further, the increment of both $T_{rtt}^i$ and $N$ increase the duration of cycles and hence, the duration of active periods for the case of OSMP-EO (refer Section \ref{sssec:trtt}). Thus, the queuing delay during active periods and hence, the average delay increases with the increment of both $T_{rtt}^i$ and $N$ as seen from Fig. \ref{fig:delay}(b) and Fig. \ref{fig:delay}(c) respectively. However, in our proposed mechanism, since the entire cycle is distributed among active ONUs, the increment of cycle time due to an increase in $T_{rtt}^i$ doesn't increase the active period significantly. Thus, the increment of average delay due an increase in $T_{rtt}^i$ is quite low as compared the OSMP-EO protocol. However, if the value of $N$ is increased then the number of active ONUs in a cycle is more which increases the duration of active periods and hence, the average delay as seen from Fig. \ref{fig:delay}(c).
 \section{Conclusion}
 In this paper, we proposed a mechanism for saving energy at ONUs in TDM-PON which can react to the instantaneous load change of both US and DS traffic at the same time. In order to do so, first, we proposed an architectural modification that allows the OLT to force ONUs waking-up from sleep mode. We then formulated an ILP for fairly distributing the active ONUs among cycles as fairly as possible while satisfying SLA of both US and DS traffic, which is followed by a wake-up scheduling protocol.
 A polynomial time $1$-approximation algorithm has been provided for solving this ILP. The convergence and the complexity analysis is also performed. Through extensive simulations, we demonstrated that the wake-up scheduling to facilitate a fair distribution of active ONUs among cycles provides a significant improvement in energy-efficiency. Further, the simulations depict that this improvement enhances with the increment of both the number of ONUs and the reach of PON network. Thus, the protocol is suitable for both dense and rural network. Also, the wake-up scheduling reduces the average delay figures. Thus, for the same delay figures, a larger value of buffer threshold can be chosen. We have also demonstrated that a larger buffer threshold  further improves energy-efficiency figures.        
\appendix
\begin{lemma}\thlabel{l1}
	If solution of $T(\mathcal{U}_b^j)$ with $\mathcal{U}_b^j=b~\forall j$ is feasible then $\sum\limits_{i=0}^{N_s-1}x^*_{ij}=n_j^*\leq b~\forall j$ where $x^*_{ij}$ is the optimal assignment.
\end{lemma}
\begin{proof}
	We prove this claim by the method of contradiction. Suppose, $x^*_{ij}$ is the optimal assignment and there exists a slot, say $j'$, such that $n_{j'}^*\geq b+1$. By Pigeonhole Principal, there must exist at lease one slot which is assigned with lower than $b$ number of ONUs. Let the set of all slots that are assigned with lower than $b$ number of ONUs as $S_u$. If we can find another assignment, $x'_{ij}$ where $|\mathbb{M}_{j'}|$ reduces by $1$ and another slot, say $j''\in S_u$, exists $\ni$  $|\mathbb{M}_{j''}|\leq b-1$ while keeping $|\mathbb{M}_{j}|$ same for all other slot then $f_1(x^*_{ij})-f_1(x'_{ij})\geq 2$. As $W>N_sM_s\sum\limits_{i,j}w_{ij}$, $f(x'_{ij})-f(x^*_{ij})>0$ which contradict our previous assumption that $x^*_{ij}$ is the optimal assignment. Now, we prove that such an  $x'_{ij}$ always exists.
	
	Clearly, if there exist an $(i,j)\in \mathcal{A}$ $\ni$ $i\in \mathbb{M}_{j'}$ and $j\in S_u$ then $x'_{ij}$ can be found by assigning ONU $i$ to slot $j$ instead of $j'$. In other word, we can say if we form a set $S_o$ by including all feasible slots of all ONUs that are assigned to slot $j'$ and $S_o\cap S_u\neq \Phi$ then we can always find a better assignment $x'_{ij}$. This process is continued for all slots of $S_o$ until no more slots are included in $S_o$. Now, it is evident that if $S_o\cap S_u\neq \Phi$ then  a better assignment, $x'_{ij}$, can always be found. Thus, the claim is invalid if $S_o\cap S_u= \Phi$. If this is true then any ONU that is assigned to a slot of $S_o$ cannot be assigned to any slot of $S_u$. Since every slot of $S_o$ are assigned with at least $b$ number of ONUs and  there exists a slot $j'$ which is assigned with more that $b$ number of ONUs, at least $b|S_o|+1$ is assigned to $|S_o|$ number of slots. Consequently, there should not exists any feasible assignment where every slots are assigned with at most $b$ number of ONUs which contradicts that solution of $T(\mathcal{U}_b^j)$ with $\mathcal{U}_b^j=b~\forall j$ is feasible.  
\end{proof}
\begin{lemma}\thlabel{l2}
	If $N_s$ number of ONU  ($\in \mathbb{N}_s$) are assigned over $M_s$ number of slot  ($\in \mathbb{S}_n$) then $f_1(x_{ij})\geq k(b-1)^2+(M_s-k)b^2$ $\forall x_{ij}$ where $b=\lceil\frac{N_s}{M_s}\rceil$ and $k=M_sb-N_s$.
\end{lemma}
\begin{proof}
	It is evident that the value $f_1(x_{ij})$ for any $\mathcal{A}$ cannot be lesser than the value $f_1(x_{ij})$ for a scenario where $\mathcal{A}$ includes all $(i,j)$ pair and hence, while proving \thref{l2}, we consider $(i,j)\in\mathcal{A}~\forall i\in \mathbb{N}_s,~j\in\mathbb{S}_n$. Here, we prove that $f_1(x_{ij})\geq k(b-1)^2+(M_s-k)b^2$ or in other words, the minimum value of $f_1(x_{ij})$ that can be achieved is when $k$ number of slots with $b-1$ number of ONUs and rest $M_s-k$ number of slots with $b$ number of ONUs where $k(b-1)+(M_s-k)b=N_s$.  We prove \thref{l2} by the method of contradiction. Suppose, the above mentioned assignment is not the best assignment of eq. (\ref{optm}). Let, the best assignment is $x'_{ij}$ which is different from the above discussed one i.e $f_1(x'_{ij})<f_1(x_{ij})~\forall x_{ij}$. By using Pigeonhole Principle, it can be claimed that there exist either a slot, say $q_1$, such that $\sum\limits_{i=0}^{N_s-1}x'_{iq_1}\leq b-2$ or a slot, say $q_2$, such that $\sum\limits_{i=0}^{N_s-1}x'_{iq_2}\geq b+1$. 
	\subsubsection*{Case $q_1$ exists}
	If $q_1$ exists then there should a at least one slot, say $q_3$, such that $\sum\limits_{i=0}^{N_s-1}x'_{iq_3}\geq b$. Now, consider an assignment say $x''_{ij}$ where one assigned ONU of slot $q_3$ is assigned to slot $q_1$ which is always possible since we consider all $(i,j)$ pair is in $\mathcal{A}$. It is easy to show that $f_1(x''_{ij})<f_1(x'_{ij})$ which contradicts our assumption that $f_1(x'_{ij})<f_1(x_{ij})~\forall x_{ij}$. 
	\subsubsection*{Case $q_2$ exists}
	If $q_2$ exists then there exist at least one slot, say $q_4$, such that $\sum\limits_{i=0}^{N_s-1}x'_{iq_4}\leq b-1$. If an ONU of slot $q_2$ is assigned to slot $q_4$ then for this new assignment, $x'''_{ij}$, $f_1(x'''_{ij})<f_1(x'_{ij})$ and it contradicts our initial consideration i.e $f_1(x'_{ij})<f_1(x_{ij})~\forall x_{ij}$.
	This proves \thref{l2}.     
\end{proof}
\begin{lemma}\thlabel{l3}
	If the assignment of $N_s$ number of ONUs ($\in \mathbb{N}_s$) in $M_s$ number of slots ($\in \mathbb{S}_n$) by solving $T(\mathcal{U}_b^j)$ with $\mathcal{U}_b^j=\lceil\frac{N_s}{M_s}\rceil=b~\forall j$ is feasible and the achieved assignment is $\hat{x}_{ij}$ then $f_1(\hat{x}_{ij})\leq\lfloor\frac{N_s}{b}\rfloor b^2+\big(N_s-\lfloor\frac{N_s}{b}\rfloor b\big)^2$.
\end{lemma}
\begin{proof}
	Similar to \thref{l2}, we prove \thref{l3} for a scenario where $(i,j)\in\mathcal{A}~\forall i\in \mathbb{N}_s,~j\in\mathbb{S}_n$. Here, we have to prove $f_1(\hat{x}_{ij})\leq\lfloor\frac{N_s}{b}\rfloor b^2+\big(N_s-\lfloor\frac{N_s}{b}\rfloor b\big)^2$ or in other words, the maximum value that $f_1(\hat x_{i,j})$ can attain by FDOS algorithm, is: $b$ number of of ONUs are assigned to $\lfloor\frac{N_s}{b}\rfloor$ number of slots and the remaining ONUs are assigned to one slot. We prove this lemma by the method of contradiction. For doing so, we assume the above mentioned assumption is not correct. Let, the worst possible assignment is $x'_{ij}$ i.e $f_1(x'_{ij})\geq f_1(x_{ij})$ $\forall x_{ij}$ and $x'_{ij}$ is different from the above mentioned assignment. In the assignment $x'_{ij}$, there must exist at least two slots, say $j'$ and $j''$, where $\sum\limits_{i=0}^{N_s}x'_{ij'}<b$ and  $\sum\limits_{i=0}^{N_s}x'_{ij''}<b$ (use Pigeonhole principal). Without loss of generality, let us assume $\sum\limits_{i=0}^{N_s}x'_{ij'}\geq \sum\limits_{i=0}^{N_s}x'_{ij''}$. Now, we consider another assignment, say $x''_{ij}$, which is achieved by assigning any one ONU of slot $j''$ to slot $j'$. It is now very easy to prove that $f_1(x''_{ij})>f_1(x'_{ij})$ which contradicts our our previous assumption that $f_1(x'_{ij})\geq f_1(x_{ij})$ $\forall x_{ij}$. This proves \thref{l3}.
\end{proof}
\begin{lemma}\thlabel{l4}
	Suppose, in FDOS, the set of slots $\mathcal{S}$, where ONUs of $\mathcal{N}$ are assigned, is divided into sets $\mathbb{L}$ and $\mathbb{O}$ respectively. If $x'_{ij}$ is a feasible assignment of the original optimization problem of eq. (\ref{optm}) and there exists a slot, say $p\in\mathbb{O}$, $\ni$ $|\mathbb{M}_p|< \dfrac{|\mathcal{N}|}{|\mathcal{S}|}=b$ then there always exists another feasible assignment $x''_{ij}$ $\ni$ $f_1(x'_{ij})-f_1(x''_{ij})\geq 2$.
\end{lemma}
\begin{proof}
	Let ${x}'_{ij}$ is a feasible assignment of eq. (\ref{optm2}). Further, let us assume $O_1$ and $O_2$ denotes the set of slots ($\in \mathbb{O}$) that are assigned with less that $b$ and more the $b$ number of ONUs respectively in $x'_{ij}$. It can be noted that every slot of $\mathbb{O}$ is assigned with $b$ number of ONUs (i.e. $\mathbb{N}'_O=b|\mathbb{O}|$). 
	As $\mathbb{N}_O=\mathbb{N}'_O\cup\mathbb{U}_s$, $|\mathbb{N}_O|\geq b|\mathbb{O}|$. 	
	In \thref{cl7}, we prove that all ONUs of $\mathbb{N}_O$ can only be assigned to slots of $\mathbb{O}$. Since, in  $x'_{ij}$, the slot $p$ is assigned with lesser that $b$ number of ONUs (i.e. $p\in O_1$), there must exist at-least one slot, say $q$, $\ni$ $|\mathbb{M}_q|>b$ (i.e. $q\in O_2$). 
	Let, the ONUs that are assigned to $O_1$ and $O_2$ are denoted by $N_{O_1}$ and $N_{O_2}$ respectively. 
	
	If there exists an $(i,j)$ pair in $\mathcal{A}$ $\ni$ $i\in N_{O_2}$ and $j\in O_1$ and in $x'_{ij}$, $i$ is assigned to slot $k$ then find $x''_{ij}$ by assigning ONU $i$ to slot $j$ and remove its assignment to $k$. Clearly, in $x''_{ij}$, $|\mathbb{M}_k|$ will reduce by one while $|\mathbb{M}_j|$ will increase by one. For all other slots exactly same number of ONUs are assigned and in this case, $f_1(x'_{ij})-f_1(x''_{ij})=2(|\mathbb{M}_k|-|\mathbb{M}_j|-1)$. Since $|\mathbb{M}_k|-|\mathbb{M}_j|\geq 2$,  $f_1(x'_{ij})-f_1(x''_{ij})\geq 2$. 
	Suppose, this is not the case but an ONU $i'$ exists, which is assigned to a slot $k'$ $\ni$ $|\mathbb{M}_{k'}|=b$ and an $(i',j')$ exists in $\mathcal{A}$ where $j'\in O_1$. 
	If there exists any ONU, say $i''\in N_{O_2}$, $\ni$ $(i'',k')\in\mathcal{A}$ then find $x''_{ij}$  by assigning $i''$ to $k'$ and $i'$ to $j'$ while removing the previous assignment of $i'$ and $i''$. Clearly, in this case also $f_1(x'_{ij})-f_1(x''_{ij})\geq 2$. Therefore, even if $k'$ is included in $O_1$, we can still claim that if there exists an $(i,j)\in \mathcal{A}$ $\ni$ $i \in N_{O_2}$ and $j\in O_1$ then we can always find another assignment, say $x''_{ij}$, $\ni$ $f_1(x'_{ij})-f_1(x''_{ij})\geq 2$. This process (similar to creation of set $\mathbb{L}$ in Algorithm 2) will continue until no slot can be included to ${O_1}$. Now, include all ONUs that are assigned to ${O_1}$ in set $N_{O_1}$ and hence, $|N_{O_1}|<b|O_1|$.  
	
	From the above discussion it is now clear that \thref{l4} doesn't hold true only if there exists no $(i,j)\in \mathcal{A}$ $\ni$ $i\in N_{O_2}$ and $j\in O_1$. If this is true then there doesn't exist any assignment $\ni$ $|M_j|\geq b ~\forall j$.  	
		However, we know that $\bar{x}_{ij}$ is an assignment where every slots of $\mathbb{O}$ and hence, $O_1$ are assigned with $b$ or more number of ONUs which contradicts our previous claim and hence, $f_1(x'_{ij})-f_1(x''_{ij})\geq 2$.     
\end{proof}
\begin{lemma}\thlabel{l5}
	Suppose, in FDOS, the set of slots $\mathcal{S}$, where ONUs of $\mathcal{N}$ are assigned, is divided into sets $\mathbb{L}$ and $\mathbb{O}$ respectively. If $x'_{ij}$ is a feasible assignment of the original optimization problem and a slot, say $p\in\mathbb{L}$, $\ni$ $|\mathbb{M}_p|> \dfrac{|\mathcal{N}|}{|\mathcal{S}|}=b$ then there always exists another feasible assignment $x''_{ij}$ $\ni$ $f_1(x'_{ij})-f_1(x''_{ij})\geq 2$.
\end{lemma}
\begin{proof}
	From the generation of set $\mathbb{L}$ it is very easy to note that we can always find $x''_{ij}$ where $|\mathbb{M}_p|$ reduces by one and another slot say $q$ exists (assigned with less than $b$ number of ONU) which is assigned with one extra ONU while keeping the number of assigned ONUs the same for all other slots. Clearly, in this case, $f_1(x'_{ij})-f_1(x''_{ij})\geq 2$ (refer \thref{l4}) which proves \thref{l5}. 
\end{proof}    
               
\bibliographystyle{IEEEtran}

\ifCLASSOPTIONcaptionsoff
  \newpage
\fi
\end{document}